\documentclass[11pt]{article}
\usepackage{amsmath,amsthm,amssymb}
\usepackage{graphicx}
\usepackage[dvipsnames]{xcolor}
\usepackage[hidelinks]{hyperref}
\definecolor{MaterialsCoral}{cmyk}{0, 0.75, 0.5, 0}
\definecolor{MaterialsSky}{cmyk}{0.6, 0, 0, 0}
\definecolor{MaterialsSun}{cmyk}{0, 0.2, 0.6, 0.05}
\definecolor{MaterialsGrass}{cmyk}{0.65, 0, 0.3, 0}

\newcommand{\Xomit}[1]{}
\usepackage{multirow}
\usepackage{doi}
\usepackage{fullpage}
\usepackage{palatino}
\usepackage{parskip}
\usepackage{enumitem}
\setlist{leftmargin=*}
\usepackage{cite}
\newcommand{\cB}{\mathcal{B}}
\newcommand{\cC}{\mathcal{C}}
\newcommand{\cF}{\mathcal{F}}
\newcommand{\cP}{\mathcal{P}}

\newtheorem{lemma}{Lemma}

\title{An Optimal Multiple-Class Encoding Scheme for a Graph of
  Bounded Hadwiger Number}

\author{Hsueh-I Lu\thanks{hil@csie.ntu.edu.tw. Department of Computer
    Science and Information Engineering, National Taiwan University,
    Taipei, Taiwan. Research supported in part by Grant
    110--2221--E--002--075--MY3 of National Science and Technology
    Council (NSTC).}}

\date{}
\begin{document}
\maketitle

\begin{abstract}
Since Jacobson [FOCS~1989] initiated the investigation of succinct
encodings for various classes of graphs 35 years ago, there has been a
long list of results on balancing the generality of the class, the
encoding and decoding speed, the succinctness of the encoded string,
and the query support.  Let $\cC_n$ denote the set consisting of the
graphs in a class $\cC$ that have at most $n$ vertices.  A class $\cC$
is \emph{nontrivial} if the information-theoretically minimum number
$\lceil\log_2 |\cC_n|\rceil$ of bits to distinguish the members of
$\cC_n$ is $\Omega(n)$.  An encoding scheme based upon a single class
$\cC$ is \emph{$\cC$-optimal} if it takes a graph $G$ of $\cC_n$ and
produces in deterministic linear time an encoded string of at most
$\log_2 |\cC_n|+o(\log_2 |\cC_n|)$ bits from which $G$ can be
recovered in linear time. Despite the extensive efforts in the
literature, trees and general graphs were the only nontrivial classes
$\cC$ admitting $\cC$-optimal encoding schemes that support the degree
query in $O(1)$ time.

Basing an encoding scheme upon a single class ignores the possibility
of a shorter encoded string using additional properties of the graph
input.  To leverage the inherent structures of individual graphs, we
propose to base an encoding scheme upon of multiple classes: An
encoding scheme based upon a family $\cF$ of classes, accepting all
graphs in $\bigcup \cF$, is \emph{$\cF$-optimal} if it is
$\cC$-optimal for each $\cC\in \cF$.  Although an $\cF$-optimal
encoding scheme is by definition $\cC$-optimal for each $\cC\in \cF$,
having a $\cC$-optimal encoding scheme for each $\cC\in\cF$ does not
guarantee an $\cF$-optimal encoding scheme.  Under this more stringent
optimality criterion, we present an $\cF$-optimal encoding scheme
$A^*$ for a family $\cF$ of an infinite number of classes such that
$\bigcup\cF$ comprises all graphs of bounded Hadwiger numbers.
Precisely, $\cF$ consists of the nontrivial quasi-monotone classes of
$k$-clique-minor-free graphs for each positive integer $k$.  Just to
name a few, examples of monotone members of $\cF$ are graphs with
genus at most $2$, $3$-colorable plane graphs, graphs of page numbers
at most $4$ and girths at least $5$, $6$-clique-minor-free graphs,
forests with diameters at most $7$, graphs having no minor that is an
$8$-cycle, and each nontrivial minor-closed class of graphs other than
the class of all graphs.  Examples of non-monotone members of $\cF$
are trees, floor-plans, triconnected planar graphs, and plane
triangulations.

Our $\cF$-optimal encoding scheme $A^*$ supports queries of degree,
adjacency, neighbor-listing, and bounded-distance shortest path in
$O(1)$ time per output.  Hence, we significantly broaden the graph
classes admitting optimal encoding schemes that also efficiently
support fundamental queries. Our $A^*$ does not rely on any
recognition algorithm or any explicit or implicit knowledge of the
exact or approximate values of $\lceil \log |\cC_n|\rceil$ for the
infinite number of classes $\cC\in \cF$. $A^*$ needs no given
embedding of the input graph.  However, $A^*$ accepts additional
information like a genus-$O(1)$ embedding, an $O(1)$-coloring, or an
$O(1)$-orientation for the input graph to be decoded from the encoded
string and answered by the query algorithms of $A^*$.
\end{abstract}

\section{Introduction}
\label{section:section1}
\paragraph{Single-class encoding schemes}
Jacobson's~\cite{Jacobson89} initial investigation 35 years ago into
compact representations of graphs has led to extensive
research~\cite{MunroN16}. Let $\cC$ be a class of graphs. An
\emph{encoding scheme} for $\cC$ consists of an encoding algorithm, a
decoding algorithm, and possibly some query algorithms.  The encoding
algorithm takes an $n$-vertex graph $G$ from $\cC$ and produces an
\emph{encoded string} $X$ from which $G$ can be recovered by the
decoding algorithm.  Let set $\cC_n$ consist of the graphs in $\cC$
having at most $n$ vertices.  Let $|S|$ denote the cardinality of set
$S$.  All logarithms throughout the paper are base $2$.  The
succinctness of $X$ is determined by comparing its bit count to the
minimum number $\lceil\log |\cC_n|\rceil$ of bits needed to
distinguish the graphs in $\cC_n$.  The encoded string is
\emph{$\cC$-succinct} if it has at most $f(n)+o(f(n))$ bits for each
continuous super-additive function $f$ with $\log |\cC_n|\leq
f(n)+o(f(n))$~\cite{HeKL00}, which essentially means that the encoding
size is bounded by $\log |\cC_n|+o(\log |\cC_n|)$.  An encoding scheme
is \emph{$\cC$-optimal} if its encoding and decoding algorithms run in
deterministic linear time and the encoded string is $\cC$-succinct.
For example, if $\cC$ is the class of simple undirected graphs, then
$\log |\cC_n|=0.5n^2-o(n^2)$~\cite{HararyP73} and hence a
$0.5n(n+1)$-bit string representing the symmetric adjacency matrix of
the input graph is the encoded string of a $\cC$-optimal encoding
scheme that supports adjacency query in constant time. If $\cC$ is the
class of rooted ordered trees, then $\log |\cC_n|=2n-O(n)$ and hence a
$2(n-1)$-bit string representing the depth-first traversal of the
input tree is the encoded string of a $\cC$-optimal encoding scheme
which can take linear time to find the parent of a vertex in the
tree~\cite{Jacobson89}.  A class $\cC$ is \emph{nontrivial} if $\log
|\cC_n|=\Omega(n)$.  Figure~\ref{figure:figure1} displays a long list
of results to balance several factors: the generality of class $\cC$,
the speed of encoding and decoding, the ability of $X$ to support
queries, and the succinctness of $X$.  Despite these extensive efforts
in the literature, trees~(see,
e.\,g.,~\cite{LuY08,MunroR01,NavarroS16}) and general graphs formed
the previously only known nontrivial classes $\cC$ of graphs admitting
$\cC$-optimal encoding schemes that also support constant-time degree
queries.

\begin{figure}\begin{center}
\scalebox{0.76}{
\renewcommand{\arraystretch}{1.2}
\begin{tabular}{|c|c|c|c|c|c|c|c|c|c|c|}
				\hline
				\multirow{2}{*}{graph classes}&\multirow{2}{*}{bits}&\multicolumn{6}{|c|}{time}&\multirow{2}{*}{references}\\
				\cline{3-8}
				&  & $\textit{adj}(u,v)$ & $\textit{nbrs}(u)$  & $\textit{deg}(u)$ & $\textit{near}(u,v)$ & encode & decode & \\
				\hline
				\multirow{7}{*}{planar} & 
				$O(n)$ & $\diamond$ & $\diamond$ & $\diamond$ & $\diamond$ & $\star$ & $\star$&\cite{BonichonGH11,HeKL99,KeelerW95,Turan84} \\
				\cline{2-9} & $O(n)$ & $O(\log n)$ & $O(d\log n)$ & $O(\log n)$ & $\diamond$ & $\bullet$ & $\bullet$ & \cite{Jacobson89} \\
				\cline{2-9} & $O(n)$ & $\star$ & $\star$ & $\star$ & $\diamond$ & $\bullet$ & $\bullet$ &\cite{GavoilleH08} \\
				\cline{2-9} & $O(n)$ & $\star$ & $\star$ & $\star$ & $\diamond$ & $\star$ & $\star$ & \cite{ChiangLL05,ChuangGHKL98,MunroR01} \\
				\cline{2-9} & $O(n)$ & $\star$ & $\star$ & $\star$ & $\diamond$ &$O(n\log n)$ & $\star$ & \cite{BlandfordBK03} \\
				\cline{2-9} & $O(n)$ & $\star$ & $\star$ & $\star$ & $\diamond$ &\Xomit{$O(n^3\log n)$}$\diamond$ & $\bullet$ & \cite{Fuentes-SepulvedaNS23}\\
				\cline{2-9} & $\star$ & $\diamond$ & $\diamond$ & $\diamond$ & $\diamond$ & $O(n \log n)$ & $O(n \log n)$ & \cite{HeKL00} \\
				\cline{2-9} & $\star$ & $\star$ & $\star$ & $\star$ & $\diamond$ & $O(n \log n)$ & $\star$ &\cite{BlellochF10} \\
				\hline
				\multirow{3}{*}{triconnected planar} & $O(n)$ & $\star$ & $\star$ & $\star$ & $\diamond$ & $\star$ & $\star$ & \cite{ChuangGHKL98} \\
				\cline{2-9} & $O(n)$ & $\diamond$ & $\diamond$ & $\diamond$ & $\diamond$ & $O(n\log n)$ & $O(n\log n)$ & \cite{HeKL00} \\
				\cline{2-9} & $\star$ & $\bullet$ & $\diamond$ & $\bullet$ & $\diamond$ & \Xomit{$O(n^3\log n)$}$\diamond$ & $\bullet$ & \cite{AleardiDS08} \\				
				\hline
				\multirow{6}{*}{plane triangulation} & $O(n)$ & $\diamond$ & $\diamond$ & $\diamond$ & $\diamond$ & $\star$ & $\star$ & \cite{HeKL99,KeelerW95} \\
				\cline{2-9} & $O(n)$ & $\star$ & $\star$ & $\star$ & $\diamond$ & $\star$ & $\star$ & \cite{ChuangGHKL98,YamanakaN10} \\
				\cline{2-9} & $\star$ & $\bullet$ & $\diamond$ & $\bullet$ & $\diamond$ & \Xomit{$O(n^3\log n)$}$\diamond$ & $\bullet$ & \cite{AleardiDS08} \\
				\cline{2-9} & $\star$ & $\diamond$ & $\diamond$ & $\diamond$ & $\diamond$ & $\diamond$ & $\diamond$ & \cite{Poulalhon2006} \\
				\cline{2-9} & $\star$ & $\diamond$ & $\diamond$ & $\diamond$ & $\diamond$ & $O(n\log n)$ & $O(n\log n)$ & \cite{HeKL00} \\
				\cline{2-9} & $\star$ & $\diamond$ & $\diamond$ & $\diamond$ & $\diamond$ & $\star$ & $\star$ & \cite{Lu14} \\
				\cline{2-9} & $\star$ & $\star$ & $\star$ & $\diamond$ & $\diamond$ & \Xomit{$O(n^3\log n)$}$\diamond$ & $\bullet$ & \cite{Fuentes-Sepulveda21} \\
				\hline
				\multirow{3}{*}{planar floorplan} & $O(n)$ & $\diamond$ & $\diamond$ & $\diamond$ & $\diamond$ & $\star$ & $\star$ & \cite{Chuang08,YamanakaN06} \\
				\cline{2-9} & $O(n)$ & $\star$ & $\star$ & $\star$ & $\diamond$ & $\star$ & $\star$ & \cite{YamanakaN10} \\
				\cline{2-9} & $\star$ & $\diamond$ & $\diamond$ & $\diamond$ & $\diamond$ & $\star$ & $\star$ & \cite{Lu14} \\
				\hline
				\multirow{2}{*}{genus-$O(1)$ triangulation} & $O(n)$ & $\star$ & $\diamond$ & $\star$ & $\diamond$ & $\star$ & $\star$ & \cite{AleardiDS2005,AleardiFL09} \\
				\cline{2-9} & $\star$ & $\diamond$ & $\diamond$ & $\diamond$ & $\diamond$ & $\star$ & $\star$ & \cite{Lu14} \\
				\hline
				\multirow{2}{*}{monotone genus-$O(1)$} & $O(n)$ & $\diamond$ & $\diamond$ & $\diamond$ & $\diamond$ & $\star$ & $\star$ & \cite{DeoL98} \\
				\cline{2-9} & $\star$ & $\diamond$ & $\diamond$ & $\diamond$ & $\diamond$ & $\star$ & $\star$ & \cite{Lu14} \\
				\hline
				genus-$O(1)$ floorplan & $\star$ & $\diamond$ & $\diamond$ & $\diamond$ & $\diamond$ & $\star$ & $\star$ & \cite{Lu14} \\
				\hline
				\multirow{2}{*}{$O(1)$-page} & $O(n)$ & $O(\log n)$ & $O(d\log n)$ & $O(\log n)$ & $\diamond$ & $\bullet$ & $\bullet$ & \cite{Jacobson89} \\
				\cline{2-9} & $O(n)$ & $\star$ & $\star$ & $\star$ & $\diamond$ & $\star$ & $\star$ & \cite{BarbayAHM12,GavoilleH08,MunroR01} \\
				\hline
				\multirow{2}{*}{hereditary separable} & $O(n)$ & $\star$ & $\star$ & $\star$ & $\diamond$ & $O(n\log n)$ & $\star$ & \cite{BlandfordBK03} \\
				\cline{2-9} & $\star$ & $\star$ & $\star$ & $\star$ & $\diamond$ & \Xomit{$O(n^3\log n)$}$\diamond$ & $\star$ & \cite{BlellochF10} \\
				\hline
				\multirow{2}{*}{minor-closed} & $O(n\log n)$ & $\star$ & $\bullet$ & $\bullet$ & $\star$ & $\star$ & $\star$ & \cite{KowalikK06} \\
				\cline{2-9} & $\star$ & $\star$ & $\star$ & $\star$ & $\diamond$ & \Xomit{$O(n^3\log n)$}$\diamond$ & $\star$ & \cite{KammerM22} \\
				\hline
				slim (including all above)& $\star$ & $\star$ & $\star$ & $\star$ & $\star$ & $\star$ & $\star$ & this paper\\
				\hline
\end{tabular}
}
\end{center}
        
\caption{We only compare the deterministic bounds for undirected
  unlabeled simple graphs in terms of the number $n$ of vertices and
  the degree $d$ of the queried vertex.  $\textit{Adj}$,
  $\textit{nbrs}$, $\textit{deg}$, and $\textit{near}$ denote the
  adjacency, neighbor-listing, degree, and bounded-distance shortest
  path queries, respectively.  Entries are marked with different
  symbols to indicate their status: A star ($\star$) indicates that
  the entry meets the optimality criterion.  A bullet ($\bullet$)
  indicates that the entry has not been analyzed but should be easy to
  extend the encoding scheme to meet the optimality criterion without
  affecting the analyzed bit counts.  A diamond ($\diamond$) indicates
  that the entry has not been analyzed and it is unclear whether the
  encoding scheme can be adjusted to meet the optimality criterion
  without affecting the analyzed bit counts.  If an analyzed entry
  does not meet our optimality criterion, we list its deterministic
  asymptotic bound.}
\label{figure:figure1}
\end{figure}

Upper-bounding the bit count of the encoded string of a graph $G$ in
$\cC_n$ by $\log |\cC_n|+o(\log|\cC_n|)$ for $\cC$-optimality is
reasonable in the worst-case scenario for all the graphs in
$\cC_n$. However, as the class $\cC$ containing $G$ becomes broader,
this $\cC$-succinctness criterion of the encoded string of $G$ becomes
looser. Encoding schemes based on a single class $\cC$ do not consider
the possibility of shortening the encoded string using the structure
of each individual graph in $\cC$. For example, the encoded string for
a $2$-colorable planar graph produced by a $\cC$-optimal encoding
scheme for the class $\cC$ of planar graphs need not be $\cB$-succinct
for the class $\cB$ of $2$-colorable planar graphs. Similarly, the
encoded string for a tree produced by a $\cC$-optimal encoding scheme
for the class $\cC$ of general graphs may not even be asymptotically
$\cB$-succinct for the class $\cB$ of trees. To encode a graph $G$ as
compactly as possible, one needs to try all known (optimal or not)
encoding schemes for all classes $\cC$ that contain $G$. This would
also require additional efforts of running recognition algorithms for
many classes on $G$.

\paragraph{Multiple-class encoding schemes}
To leverage the inherent structures within individual graphs, we
propose to base an encoding scheme upon a family $\cF$ of multiple
graph classes. The input $n$-vertex graph $G$ is taken from the ground
class (i.\,e., the union $\bigcup \cF$ of all member classes) of the
family $\cF$.  To assess the succinctness of the encoded string of an
$n$-vertex graph $G$, we compare its bit count to the minimum of $\log
|\cC_n|$ over all classes $\cC$ that satisfy $G\in \cC\in \cF$.  The
encoded string $X$ for $G$ is \emph{$\cF$-succinct} if it is
$\cC$-succinct for all classes $\cC$ with $G\in \cC\in \cF$, meaning
that the bit count of $X$ is at most $f(n)+o(f(n))$ for every
continuous super-additive function $f$ that satisfies $\log
|\cC_n|\leq f(n)+o(f(n))$ for at least one class $\cC$ with $G\in
\cC\in \cF$.  For instance, consider the family $\cF$ of the classes
$\cC^{(k)}$ of $k$-colorable graphs. An $\cF$-succinct encoded string
of a 3-colorable non-bipartite graph is $\cC^{(k)}$-succinct for each
$k\geq 3$ but need not be $\cC^{(2)}$-succinct.  Hence, the
$\cF$-succinctness criterion of the encoded string for $G\in \cC\in
\cF$ is at least as stringent as the $\cC$-succinctness criterion,
even when $\log|(\bigcup \cF)_n|$ is considerably larger than
$\log|\cC_n|$.  An encoding scheme is \emph{$\cF$-optimal} if it is
$\cC$-optimal for every member class $\cC$ of the family $\cF$. In
other words, for an encoding scheme to be $\cF$-optimal, its encoding
algorithm must run in deterministic linear time on every graph $G\in
\bigcup \cF$ to generate an $\cF$-succinct encoded string $X$ from
which its decoding algorithm should recover $G$ in deterministic
linear time.  As the diversity of member classes in $\cF$ increases,
an $\cF$-optimal encoding scheme exploits a broader range of graph
structures.

Having a $\cC$-optimal encoding schemes for each member class $\cC$ in
a family $\cF$ does not guarantee an $\cF$-optimal encoding scheme, as
finding a class $\cC$ in $\cF$ containing $G$ with the minimum
$|\cC_n|$ can be expensive. For instance, let the $k$-th class in
$\cF$ consists of the $k$-colorable graphs.  A collection of
$\cC$-optimal encoding schemes for all classes $\cC\in \cF$ need not
yield an $\cF$-optimal encoding scheme, since computing the chromatic
number of a graph is NP-complete~\cite{GareyJ79}.  Hence, it might
appear impossible to design an $\cF$-optimal encoding scheme unless
determining a minimizer of $|\cC_n|$ over all classes $\cC$ with $G\in
\cC\in \cF$ takes linear time. However, our result in this paper
indicates that this is not necessarily the case. In what follows, we
first provide some definitions for graphs and their classes and then
explain our $\cF$-optimal encoding scheme $A^*$ for a family $\cF$ of
an infinite number of graph classes such that $\bigcup \cF$ comprises
all graphs of bounded Hadwiger numbers.  Our $\cF$-optimal encoding
scheme $A^*$ does not require explicit or implicit knowledge of the
exact or approximate values of $\log |\cC_n|$ for the member classes
$\cC$ of $\cF$.  Our $A^*$ needs no recognition algorithm of any
member class $\cC$ of $\cF$.  $A^*$ is $\cC$-optimal for each class
$\cC$ in $\cF$ even if recognizing a graph of some class $\cC\in \cF$
is an undecidable problem.

\paragraph{Concepts for graphs}
An \emph{induced subgraph} of $G$ is a graph that can be obtained by
deleting zero or more vertices and their incident edges from $G$. A
\emph{subgraph} of $G$ is a graph that can be obtained from an induced
subgraph of $G$ by deleting zero or more edges. A \emph{minor} of $G$
is a graph that can be obtained from a subgraph of $G$ by contracting
zero or more edges. An induced subgraph (respectively, a subgraph) of
a graph $G$ is a subgraph (respectively, a minor) of $G$, but not the
other way around.

Let $V(G)$ denote the vertex set of a graph $G$, and $E(G)$ its edge
set.  Let $\vec{uv}$ denote a directed edge from $u$ to $v$.  Let
$G^r$ for a graph $G$ denote the graph on $V(G)$ with $\vec{uv}\in
E(G)$ if and only if $\vec{vu}\in E(G^r)$.  Let $G\cap H$ denote the
maximal common subgraph of graphs $G$ and $H$.  Let $G\cup H$ denote
the minimal common supergraph of graphs $G$ and $H$.  The
\emph{Hadwiger number}~\cite{Kostochka84} (also known as the
\emph{contraction clique number}~\cite{BollobasCE80}) of a graph $G$,
denoted $\eta(G)$, is defined as the largest integer $k$ such that the
$k$-clique is a minor of $G\cup G^r$.  Thus, the Hadwiger number of a
forest or a planar graph is at most $2$ or $4$,
respectively. Hadwiger's conjecture~\cite{Hadwiger43} states that the
Hadwiger number $\eta(G)$ upper-bounds the chromatic number of $G$,
and it remains one of the deepest~\cite{BollobasCE80} and most
famous~\cite{Seymour16} open problems in graph theory.

An edge $\vec{uv}$ is called an \emph{outgoing} edge of $u$
(respectively, an incoming edge of $v$) or simply a \emph{$u$-out}
(respectively, $v$-in) edge.  An \emph{$i$-orientation} is a graph in
which each vertex has at most $i$ outgoing edges.  A graph $D$ is an
$i$-orientation \emph{for} a graph $G$ if $D$ is an $i$-orientation
with $D\cup D^r = G\cup G^r$.  It is known that
$|E(G)|=O(|V(G)|\cdot\eta(G)\sqrt{\log \eta(G)})$~(see,
e.g.,~\cite{Kostochka84,Mader67,Thomason01}). Since $\eta(G)=O(1)$
implies $|E(G)|=O(|V(G)|)$, the minimum degree of a graph $G$ with
$\eta(G)=O(1)$ is $O(1)$.  By $\eta(H)\leq \eta(G)$ for each induced
subgraph $H$ of $G$, an $n$-vertex graph $G$ with $\eta(G)=O(1)$
admits an $O(n)$-time obtainable $O(1)$-orientation.  Although the
Hadwiger number $\eta(G)$ of an $n$-vertex graph $G$ is NP-hard to
compute~\cite{Eppstein09} and cannot be computed in $n^{o(n)}$
time~\cite{FominLMSZ21} unless ETH~\cite{ImpagliazzoP99} fails, one
can determine whether $\eta(G)\leq h$ for any given $h=O(1)$ in
$O(n^2)$ time~\cite{KawarabayashiKR12,RobertsonS95-graphminor13} and
possibly in $O(n\log n)$ time~\cite{DujmovicHJRW13}.  For related work
on graphs with bounded Hadwiger numbers,
see~\cite{BravermanJKW21,CzerwinskiNP21,FoxW17,Grohe12,GroheNW23,GroheS20,LeviR15,LokshtanovPP022}.

\paragraph{Concepts for graph classes}
We defined a class $\cC$ of graphs as nontrivial if $\log
|\cC|=\Omega(n)$.  The class of unlabeled undirected paths is not
nontrivial, while those of trees and plane triangulations are
nontrivial since the logarithms of the numbers of $n$-vertex trees and
plane triangulations are $2n-o(n)$ and $(\log \frac{256}{27})\cdot
n-o(n)$~\cite[(8.1)]{Tutte62}, respectively.  The intersection of
nontrivial classes may not be nontrivial since it can be empty.

A class $\cC$ of graphs is \emph{minor-closed}, \emph{monotone}, or
\emph{hereditary} if every minor, subgraph, or induced subgraph of
each graph in $\cC$ remains in $\cC$, respectively. A minor-closed
(respectively, monotone) class of graphs is also monotone
(respectively, hereditary), but not vice versa. For instance, the
class of trees is not hereditary, the class of complete graphs is
hereditary but not monotone, the class of 2-colorable graphs is
monotone but not minor-closed, and the class of forests is
minor-closed.  The intersection of minor-closed, monotone, or
hereditary classes remains minor-closed, monotone, or hereditary,
respectively.

Let $G[U]$ (respectively, $G-U$) for a vertex subset $U$ of $G$ denote
the subgraph of $G$ induced by $U$ (respectively, $V(G)\setminus U$).
Let the $N_G(U)$ of a vertex subset $U$ in $G$ consist of the vertices
of $V(G-U)$ that is adjacent to one or more vertices of $U$ in $G$.
Disjoint vertex sets $U$ and $V$ are \emph{adjacent} in a graph $G$ if
$N_G(U)\cap V\ne\varnothing$, i.\,e., $\{\vec{uv},\vec{vu}\}\cap
E(G)\ne\varnothing$ holds for some vertices $u\in U$ and $v\in V$.
Let $N_G[U]=N_G(U)\cup U$. The \emph{open} (respectively,
\emph{closed}) \emph{neighborhood} of $U$ in $G$ is $G[N_G(U)]$
(respectively, $G[N_G[U]]$).  The \emph{quasi-neighborhood} $G(U)$ of
$U$ in $G$ is the closed neighborhood of $U$ in $G$ excluding the
edges in the open neighborhood of $U$ in $G$, i.\,e.,
$G(U)=G[N_G[U]]\setminus E(G-U)$.  A \emph{$k$-quasi-member} of a
class $\cC$ of graphs is a graph that can be obtained from a graph in
$\cC$ by deleting at most $k$ vertices and their incident edges.  A
class $\cC$ of graphs is \emph{quasi-monotone} if the
quasi-neighborhood $G(U)$ of each vertex subset $U$ of each graph $G$
in $\cC$ is an $O(|N_G(U)|)$-quasi-member of $\cC$.  For instance, a
forest can be made a tree by adding a new vertex with several incident
edges to connect its connected components.  Thus, the class of trees
is quasi-monotone, since each non-tree subgraph of a tree is a
$1$-quasi-member of the class of trees.  A monotone class of graphs is
quasi-monotone but not vice versa.  The intersection of quasi-monotone
classes remains quasi-monotone.

We call a class $\cC$ of graphs \emph{slim} if it is nontrivial and
quasi-monotone and admits a constant that upper-bounds the Hadwiger
numbers of all graphs in $\cC$.  Examples of slim classes include
trees, forests, series-parallel graphs, and all classes listed in
Figure~\ref{figure:figure1}. The monumental graph minor
theory~\cite{Lovasz06} by Robertson and
Seymour~\cite{RobertsonS83-graphminor1, RobertsonS86-graphminor2,
  RobertsonS84-graphminor3, RobertsonS90-graphminor4,
  RobertsonS86-graphminor5, RobertsonS86-graphminor6,
  RobertsonS88-graphminor7, RobertsonS90-graphminor8,
  RobertsonS90-graphminor9, RobertsonS91-graphminor10,
  RobertsonS94-graphminor11, RobertsonS95-graphminor12,
  RobertsonS95-graphminor13, RobertsonS95-graphminor14,
  RobertsonS96-graphminor15, RobertsonS03-graphminor16,
  RobertsonS99-graphminor17, RobertsonS03-graphminor18,
  RobertsonS04-graphminor19, RobertsonS04-graphminor20,
  RobertsonS09-graphminor21, RobertsonS12-graphminor22,
  RobertsonS10-graphminor23} confirms Wagner's~\cite{Wagner70}
conjecture that each minor-closed class of undirected graphs can be
defined by a finite list of excluded minors. Hence, all nontrivial
minor-closed classes of graphs other than the one composed of all
graphs (whose forbidden minor set is empty) are slim, implying that
the class of graphs with genus no more than a constant $g$ is
slim. Moreover, since each graph having a bounded Hadwiger number is
separable~\cite{AlonST94,ReedW09,Wulff-Nilsen11}, an $n$-vertex graph
in a slim class $\cC$ of graphs can be represented in $O(n)$
bits~\cite{BlandfordBK03}. Therefore, $\log |\cC_n|=\Theta(n)$ holds
for each slim class $\cC$ of graphs.

\paragraph{Our optimal multiple-class encoding scheme}
For the rest of the paper, let $\cF$ denote the family of all slim
classes of graphs.  We present an $\cF$-optimal encoding scheme
$A^*$. This means that our encoding scheme $A^*$ is $\cC$-optimal for
every nontrivial quasi-monotone class $\cC$ of graphs that admits a
constant $h$ satisfying the condition $\eta(G)\leq h$ for all graphs
$G\in \cC$.  Since for each integer $h\geq 2$ the graphs $G$ with
$\eta(G)\leq h$ form a distinct slim class of graphs, the number of
member classes within $\cF$ is infinite.  Imposing other known or
unknown nontrivial properties that are monotone or quasi-monotone on
each of these infinite slim classes leads to more varieties.  Examples
of monotone slim classes include graphs with genus at most $2$,
$3$-colorable plane graphs, graphs with page number at most $4$, and
graphs with girths at least $5$. Examples of non-monotone slim classes
include trees, triconnected planar graphs, and triangulations or
floor-plans of genus-$O(1)$
surfaces~(see~\S\ref{subsection:subsection2.4}).
 
We present our $\cF$-optimal encoding scheme~$A^*$ using a simple
unweighted directed $n$-vertex graph $G$ given in an adjacency list.
Our~$A^*$ accepts additional information of $G$ such as vertex colors
or edge directions which can be recovered in tandem with $G$ by the
decoding algorithm of $A^*$ and answered by the query algorithms of
$A^*$, as long as at least one slim class $\cC$ containing $G$
equipped with the additional information satisfies $\log
|\cC_n|=\Theta(n)$.  Moreover, if the given adjacency list of $G$
reflects a genus-$O(1)$ embedding, then $A^*$ can also accept the
embedding as additional information.  Our $A^*$ supports the following
fundamental queries in $O(1)$ time:
\begin{itemize}
\item
Output $|N_G(u)|$ and the numbers of $u$-out and $u$-in edges of $G$.

\item
Output a neighbor $v$ of $u$ in $G$, if $|N_G(u)|\geq 1$.

\item
Output the color assigned to the vertex $u$.

\item
Output whether $\vec{uv}\in E(G)$.
  
\item
Output the direction assigned to the edge $\vec{uv}$ by an equipped
orientation $D$ for $G$.
  
\item
Output the incident edges of $u$ (respectively $v$) preceding and
succeeding $\vec{uv}$ in clockwise orders around $u$ (respectively,
$v$) according to the embedding of $G$ produced by the decoding
algorithm.

\item
Output a shortest $uv$-path of $G$ if there is one with length bounded
by a prespecified constant $t$.
\end{itemize}
The neighbor-listing query for a vertex $u$ can be supported in
$O(|N_G(u)|)$ time using the $O(1)$-time queries of reporting a
neighbor and the next neighbor.  The above list is not
exhaustive. Additional queries can be created by following the design
of the $o(n)$-bit encoded strings and $O(1)$-time query algorithms
elaborated in~\S\ref{section:section3}.  As mentioned above,
$\cC$-optimal encoding schemes that support $O(1)$-time degree query
were only known for the classes $\cC$ of trees and general
graphs. Hence, our $\cF$-optimal encoding scheme $A^*$ significantly
broadens the graph classes $\cC$ that admit $\cC$-optimal encoding
schemes which also support fundamental queries in $O(1)$ time per
output.

The ground class $\bigcup \cF$ contains all graphs with bounded
Hadwiger numbers, since the class of graphs with $\eta(G)\leq h$ for
each integral constant $h\geq 2$ is slim. Thus, the encoding algorithm
of $A^*$ can compute an encoded string $X$ for an input $n$-vertex
graph $G$ with $\eta(G)=O(1)$ in $O(n)$ time, and the decoding
algorithm of $A^*$ can decode $X$ back to $G$ in $O(n)$ time. The bit
count of $X$ is essentially bounded by $\log |\cC_n|+o(\log |\cC_n|)$
for each slim class $\cC$ containing $G$, which means that our
encoding scheme $A^*$ automatically exploits all slim structures of
$G$ to encode $G$ as compactly as possible. We emphasize that $A^*$
does not require explicit or implicit knowledge of the slim structures
$\cC$ of $G$ or the exact or approximate values of their $\lceil\log
|\cC_n|\rceil$.

The encoded string $X$ for an input $n$-vertex graph $G\in \bigcup
\cF$ produced by our encoding scheme $A^*$ is the concatenation of an
$\cF$-succinct base string $X_{\textit{base}}$ and an $o(n)$-bit
string $X_q$ for each supported query~$q$. The input graph $G$ and the
equipped additional information can be reconstructed using solely
$X_{\textit{base}}$.  Thus, the bit count of $X_{\textit{base}}$ is
affected by the complexity of the additional information.  As an
example, $A^*$ can encode a $3$-colored undirected plane graph $G$ as
a directed planar graph that is equipped with a $3$-coloring for the
vertices and a planar embedding for the edges according to which the
edge $\vec{uv}$ for every two adjacent vertices $u$ and $v$ of $G$
always immediately precedes the edge $\vec{vu}$ in clockwise order
around $u$.

\begin{itemize}
\item 
If the coloring and embedding are required to be recovered together
with $G$ by the decoding algorithm of $A^*$, then the base string
$X_{\textit{base}}$ is $\cC$-succinct for each slim class $\cC$ of
$3$-colored undirected plane graphs that contains the input graph $G$.
This $X_{\textit{base}}$ need not be $\cC$-succinct for each slim
class $\cC$ of $3$-colorable undirected planar graphs (i.\,e., without
equipped planar embeddings) that contains $G$.  Via two $o(n)$-bit
strings $X_q$ appended to $X_{\textit{base}}$, the query algorithms of
$A^*$ support the two queries $q$ of reporting the color of a vertex
and obtaining the edges immediately succeeding and preceding an edge
$\vec{uv}$ around $u$ (respectively, $v$) in clockwise order according
to the equipped planar embedding of $G$.

\item 
If the given coloring and planar embedding of $G$ need not be
recovered, the decoding algorithm of $A^*$ reports a graph $H$
isomorphic to $G$ represented in an adjacency list reflecting an
embedding of $H$ that might have nonzero genus.  The base string
$X_{\textit{base}}$ is $\cC$-succinct for each slim class $\cC$ of
$3$-colorable undirected planar graphs that contains the input graph
$G$.  $A^*$ supports the query of obtaining the edges immediately
succeeding and preceding an edge $\vec{uv}$ around $u$ (respectively,
$v$) in clockwise order according to the embedding of $H$.
\end{itemize}
The concatenation structure of $X$ allows new queries be supported by
appending an $o(n)$-bit string $X_{q}$ for each new query $q$ to the
original $X$, eliminating the need to recompute the entire encoded
string.

\paragraph{Comparing with previous work} 
Reed and Wood~\cite{ReedW09} suggested that the most general setting
for separator theorems involves graphs that exclude a fixed minor. It
remains unclear whether a graph's separability implies a constant
bound on its Hadwiger number. However, all hereditary classes of
separable graphs known to date are
minor-closed~\cite{KawarabayashiR10}. Consequently, our $\cF$-optimal
encoding scheme $A^*$ outperforms all previous encoding schemes listed
in Figure~\ref{figure:figure1}.
\begin{itemize}
\item
The encoded strings produced by He, Kao, and Lu's~\cite{HeKL00}
encoding scheme for the class $\cC$ of planar graphs and Blandford,
Blelloch, and Kash's~\cite{BlandfordBK03} encoding scheme for
hereditary classes $\cC$ of separable graphs are $\cC$-succinct.
However, their encoding and decoding algorithms run in $O(n\log n)$
time. Lu~\cite{Lu14} extended He et al.'s framework to accommodate all
non-trivial monotone classes of bounded-genus graphs and some
non-monotone classes of graphs, improving the encoding and decoding
time to $O(n)$, but without supporting efficient queries.  Blelloch
and Farzan~\cite{BlellochF10} extended Blandford et
al.'s~\cite{BlandfordBK03} encoding framework to support adjacency,
degree, and neighbor queries in $O(1)$ time per output. Blelloch et
al.'s extension is based on Raman, Raman, and Satti's~\cite{RamanRS07}
indexable dictionary~\cite[Lemma~1]{BlellochF10} and fully indexable
dictionary~\cite[Lemma~2]{BlellochF10}.  Raman et
al.~\cite[\S8]{RamanRS07} only mentioned that their indexable
dictionary can be constructed in expected linear time without
commenting on its deterministic time bound.  The expected and
deterministic time bounds of their fully indexable dictionary are not
explicitly stated.  Let $\textit{poly}(f(n))=f(n)^{O(1)}$.  According
to the first author Raman~\cite{Raman02} of the indexable
dictionary~\cite{RamanRS07}, the deterministic constructing time of an
indexable dictionary for a $\textit{poly}(n)$-bit string having $n$
$1$-bits is $O(n^3\log n)$, which is the same as that of Fredman,
Koml{\'{o}}s, and Szemer{\'{e}}di's~\cite{FredmanKS84} data structure.

\item
Previous encoding
schemes~\cite{ChiangLL05,ChuangGHKL98,Fuentes-SepulvedaNS23,GavoilleH08,KowalikK06,MunroR01}
for the classes $\cC$ of planar and $O(1)$-page graphs, which support
adjacency, degree, neighbors, or bounded-distance shortest path
queries in $O(1)$ time per output, take linear time to encode and
decode, but their encoded string are not $\cC$-succinct. In
particular, Kowalik and Kurowski's~\cite{KowalikK06} encoding scheme
for planar graphs, which is extendable to accommodate minor-closed
classes, requires $O(n)$ words and thus $O(n\log n)$ bits.  Kammer and
Meintrup~\cite{KammerM22} optimized the bit count for minor-closed
classes but did not support efficient query of bounded-distanced
shortest path.  Kammer et al.'s result is based on Blelloch and
Farzan's~\cite{BlellochF10} encoding framework for hereditary classes
of separable graphs, so their encoding algorithm runs in deterministic
$O(n^3\log n)$ time.  Fuentes-Sep\'{u}lveda, Seco, and
Via\~{n}a's~\cite{Fuentes-Sepulveda21} result for plane triangulations
and Fuentes-Sep\'{u}lveda, Navarro, and
Seco's~\cite{Fuentes-SepulvedaNS23} result for planar graphs also use
Raman et al.'s~\cite{RamanRS07} indexable dictionary.
Castelli~Aleardi, Devillers, and Schaeffer~\cite{AleardiDS08}
mentioned that their encodings for triconnected planar graphs and
plane triangulations supporting local queries can be constructed in
linear time.  However, they did not provide details about how to
answer queries for concatenated code segments. Assuming that their
encoding algorithm is also based on Raman et al.'s indexable
dictionary, their deterministic encoding time is $O(n^3\log n)$.
\end{itemize}

Our $\cF$-optimal encoding scheme $A^*$ has only one encoding
algorithm that automatically handles all infinite member classes $\cC$
in $\cF$, regardless of whether a class $\cC$ is monotone or merely
quasi-monotone.  That is, there is no need to adjust our $A^*$ for any
non-monotone member class $\cC$ of $\cF$.  Moreover, $A^*$ does not
require a user to supply any recognition algorithms for any slim class
$\cC$ of graphs with prespecified time complexity bounds. Note that
$A^*$ is effective even if the recognition problem for a member class
$\cC$ of $\cF$ is undecidable.  In contrast, each previous
$\cC$-succinct encoding
scheme~\cite{BlellochF10,AleardiDS08,HeKL00,KammerM22,Lu14,PoulalhonS06}
for a non-trivial class $\cC$ of graphs other than trees and general
graphs requires its user to specify $\cC$ by supplying a recognition
algorithm for the class $\cC$ that runs in $O(1)^{\textit{poly}(n)}$
time in order to construct a look-up table for all tiny graphs in
$\cC$. Moreover, different classes $\cC$ need different approaches for
their decoding algorithm to glue tiny subgraphs of the input graph $G$
to recover $G$. For instance, Castelli Aleardi et
al.~\cite{AleardiDS08} used two separate sections of their paper to
describe their encoding schemes for plane triangulations and
triconnected planar graphs. Lu's~\cite[\S5]{Lu14} framework for
non-monotone classes accommodated triangulations and floor-plans in
different ways.

\paragraph{Technical overview}
Similar to all former $\cC$-succinct encoding schemes for nontrivial
classes $\cC$ other than trees and general graphs, our $\cF$-optimal
encoding scheme $A^*$ is based on decomposing the input graph into
tiny subgraphs.  Lu's~\cite{Lu14} encoding algorithm on a graph whose
bounded-genus embedding is either given or linear-time computable uses
Dijdjev and Venkatesan's~\cite{DjidjevV95} planarizer for a graph of
bounded genus and Goodrich's~\cite{Goodrich95} separator-decomposition
tree for a planar graph.  Our encoding scheme $A^*$ decomposes the
input graph based on Reed and Wood's algorithms to obtain a separator
of $G$~\cite[Theorem~1.2]{ReedW09} and an $H$-partition of $G$ for an
$O(n/\log^2 n)$-vertex graph $H$ with $\eta(H)\leq
\eta(G)$~\cite[Lemma~4.1]{ReedW09}.  By choosing these tools, our
encoding scheme $A^*$ does not rely on any embedding of the input
graph, although it is acceptable to encode a genus-$O(1)$ embedding in
tandem with the input graph.

To avoid the requirement of supplying a recognition algorithm for any
slim class $\cC$, our encoding scheme $A^*$ does not precompute a
look-up table for all tiny subgraphs in the slim classes
$\cC$. Instead, $A^*$ collects the tiny subgraphs of the input graph
$G$ at the bottom of the recursive decomposition and computes a
code-book string only for them. We will prove that, as long as $G$
belongs to a slim class $\cC$ (no matter whether $\cC$ is monotone or
merely quasi-monotone), the encoded string is guaranteed to be
$\cC$-succinct.  Our encoding scheme $A^*$ can be directly applied to
any graph having a bounded Hadwiger number without modification for
any slim class.  For a given class $\cC$ of graphs to have our
encoding scheme $A^*$ as a $\cC$-optimal encoding scheme, one just
have to prove that $\cC$ is slim.  This is exactly what we do for
triconnected planar graphs and triangulations and floor-plans of
genus-$O(1)$ surfaces in~\S\ref{subsection:subsection2.4}.  In order
to ensure that the encoding length we get is also $\cC$-succinct for
each non-monotone slim class $\cC$ that contains the input graph $G$,
we classify a $k$-vertex subgraph $H$ of $G$ into two sets according
to whether it admits a vertex subset $U\subseteq V(H)$ with $H=G(U)$
and $|N_G(U)|\leq |V(H)|/\log^2 |V(H)|$.  If it does, then we can
prove that the encoded string for $H$ is $\cC$-succinct by showing
that its bit count is bounded by $f(k)+o(f(k))$ for any any continuous
super-additive function $f$ with $\log|\cC_k|\leq f(k)+o(f(k))$.
Although the encoded string for an $H$ that does not admit such a
vertex subset $U$ need not be $\cC$-succinct, we manage to show that
the number of such subgraphs $H$ that occur in our recursive encoding
algorithm is too few to ruin the $\cC$-succinctness of the final
encoded string for $G$.  The ``star partition''
(see~\S\ref{subsection:subsection2.1}) is crucial in ensuring this.

The key of supporting queries lies in designing an $o(n)$-bit string
so that information involving multiple tiny graphs can be obtained in
$O(1)$ time.  Our $\cF$-optimal encoding scheme $A^*$ does not rely on
an indexable dictionary like that of Raman et al.~\cite{RamanRS07}
whose encoding size is very close to information-theoretical minimum,
since we cannot afford its super-linear deterministic construction
time.  Instead, we just need an $O(r\log m)+o(m)$-bit representation
for an $m$-bit string having $r$ $1$-bits to support its membership,
rank, and select queries in $O(1)$ time.  We show
in~\S\ref{subsection:subsection3.1} that such a representation can be
constructed in deterministic linear time (see
Lemma~\ref{lemma:lemma4}).

A challenging task for the query support of $A^*$ is to design an
$O(n)$-time computable $o(n)$-bit string $X_q$ from which the query
$\textit{near}_t(u,v)$ of bounded-distance shortest path can be
answered in $O(1)$ time.  Kowalik and Kurowski~\cite{KowalikK06}
showed an involved $O(n)$-time computable $O(n\log n)$-bit data
structures for an $n$-vertex undirected planar graph that supports the
query in $O(1)$ time.  They only briefly sketched how the data
structure can be extended to work for minor-closed classes of graphs.
Our improvement is two-fold:
\begin{itemize}
\item
We provide a simpler solution to achieve the same objective of
supporting the query in $O(1)$ time not just for an undirected planar
graph but directly for an $n$-vertex directed graph $G$ with
$\eta(G)=O(1)$.  Specifically, we present an $O(n)$-time computable
$O(1)$-orientation $D$ for $G$ from which the query can be answered in
constant time.  The description and proof for such a $D$, which is
called a ``director'' for $G$ in \S\ref{subsection:subsection3.4}, are
much shorter than those of Kowalik and Kurowski's data structure.
Recall that our $\cF$-optimal encoding scheme $A^*$ accepts additional
information like an orientation for the graph.  Thus, $A^*$ can
directly encode $G$ together with a director $D$ for $G$ such that the
direction of each edge of $G$ assigned by $D$ can be answered in
$O(1)$ time.  This leads to an $O(n)$-time computable $O(n)$-bit
encoded string for $(G,D)$ that supports the bounded-distance
shortest-path query in $O(1)$ time, already improving upon the
$O(n\log n)$-bit result of Kowalik and Kurowski for the class of
planar graphs and minor-closed classes of graphs.

\item 
The above $O(n)$-bit encoded string for a pair $(G,D)$ of two graphs
$G$ and $D$ need not be an $\cF$-succinct encoded string for $G$.  To
obtain an $\cF$-succinct encoded string for $G$ that also supports the
query in $O(1)$ time, we show that $G$ admits a director $D$ such that
the number of extra bits required to encode $(G,D)$ in addition to
that required to encode $G$ is only $o(n)$.  The trick is to directly
equip a director for each decomposed tiny subgraph at the bottom of
the recursive decomposition and keep it in an $o(n$)-bit code-book
string. As a result, we do not need too many extra bits to encode the
equipped director for each tiny subgraph.  The challenge then lies in
showing that $o(n)$ bits suffice to combine the equipped directors of
those tiny subgraphs of $G$ into a director for $G$.  The star
partition of $G$ to be presented in~\S\ref{subsection:subsection2.1}
is once again crucial in accomplishing this objective.
\end{itemize}

\paragraph{Computation model and road-map}
We assume the conventional unit-cost RAM model of computation, in
which operations such as read, write, and add on $O(\log n)$
consecutive bits take $O(1)$ time.  The model has been adopted by all
previous work on graph encoding except that of
Jacobson~\cite{Jacobson89}.

The rest of the paper is organized as
follows. Section~\ref{section:section2} shows the linear-time encoding
and decoding algorithms of $A^*$.  Section~\ref{section:section3}
shows the query algorithms of $A^*$, which run in $O(1)$ time per
output.  Section~\ref{section:section4} concludes the paper.

\section{The encoding and decoding algorithms} 
\label{section:section2}

Section~\ref{subsection:subsection2.1} shows that each graph of
bounded Hadwiger number admits a linear-time computable star
partition.  Sections~\ref{subsection:subsection2.2}
and~\ref{subsection:subsection2.3} present the encoding and decoding
algorithms of our $\cF$-optimal encoding scheme $A^*$ based on star
partition. Section~\ref{subsection:subsection2.4} shows that the
classes $\cC$ of triconnected planar graphs and triangulations and
floor-plans of genus-$O(1)$ surfaces are all slim, implying that $A^*$
is $\cC$-optimal for each of these three classes $\cC$.

\subsection{Star partition}
\label{subsection:subsection2.1}
Let $[i,j]$ for integers $i$ and $j$ consist of the integers $k$ with
$i\leq k\leq j$.  Let $[j]=[1,j]$.  We call $(V_0,\ldots,V_p)$ a
\emph{partition} of a graph $G$ if $V_0,\ldots,V_p$ are
pairwise-disjoint subsets of $V(G)$ whose union is $V(G)$.  As
illustrated in Figure~\ref{figure:figure2}(a), a partition
$(V_0,\ldots,V_p)$ of an $n$-vertex graph $G$ is a \emph{star}
partition of $G$ if
\begin{enumerate}[label=\emph{S\arabic*:}, ref={S\arabic*}]
\item 
\label{condition:S1}
$V_1,\ldots,V_p$ are pairwise nonadjacent in $G$,

\item 
\label{condition:S2}
$|N_G[V_i]|=\textit{poly}(\log n)$ holds for each $i\in[p]$, and

\item 
\label{condition:S3}
$|V_0|+p+\sum_{i\in [p]} |N_G(V_i)|=O(n/\log^{4/3} n)$.
\end{enumerate}

\begin{figure}
\centerline{\scalebox{0.5}{\includegraphics{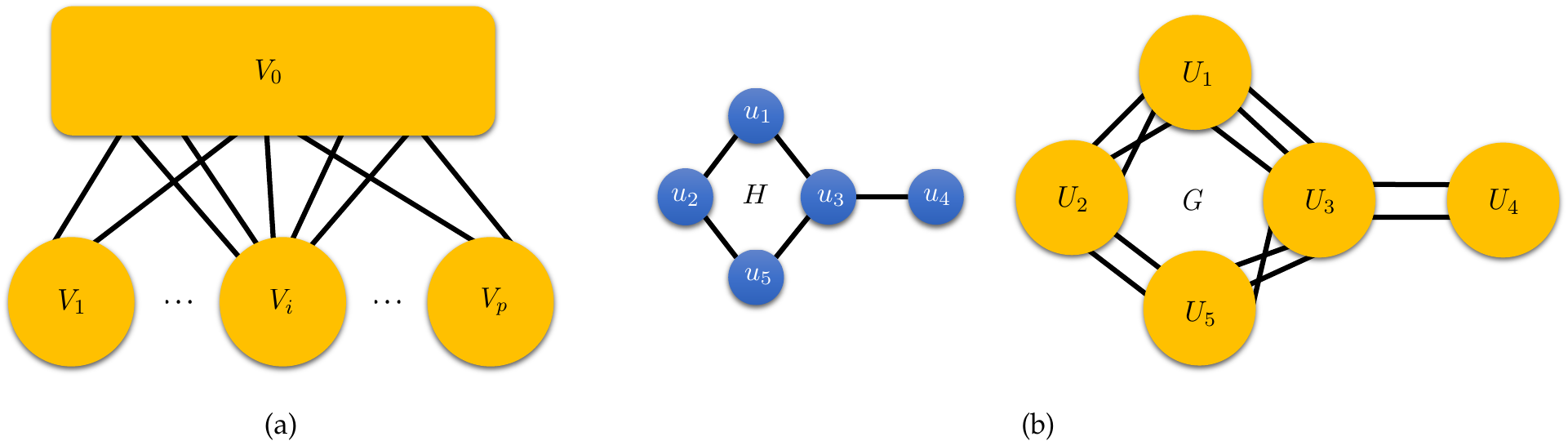}}}
\caption{(a) An illustration for star partition. (b) An illustration
  for an $H$-partition of $G$.}
\label{figure:figure2}
\end{figure}

\begin{lemma}
\label{lemma:lemma1}
It takes $O(n)$ time to obtain a star partition of an $n$-vertex graph
$G$ with $\eta(G)=O(1)$.
\end{lemma}

The rest of the subsection proves Lemma~\ref{lemma:lemma1} by
Lemmas~\ref{lemma:lemma2} and~\ref{lemma:lemma3} below.  As
illustrated in Figure~\ref{figure:figure2}(b), a partition
$(U_1,\ldots,U_m)$ of a graph $G$ is an \emph{$H$-partition} of $G$
for an undirected graph $H$ with $V(H)=\{U_1,\ldots,U_m\}$ if the next
statements hold for any distinct indices $i$ and $j$ in $[m]$:
\begin{enumerate}[label=\emph{H\arabic*:}, ref={H\arabic*}, leftmargin=*]
\item 
\label{condition:K1}
$|U_i|=O(n/m)$.

\item
\label{condition:K2}
$U_iU_j$ is an edge of $H$ if and only if the subsets $U_i$ and $U_j$
of $V(G)$ are adjacent in $G$.
\end{enumerate}

\begin{lemma}[{Reed and Wood~\cite[Lemma~4.1]{ReedW09}}]
\label{lemma:lemma2}
For an $n$-vertex graph $G$ with $\eta(G)=O(1)$ and an $m=O(n)$, it
takes $O(n)$ time to obtain an $H$-partition of $G$ for an
$O(m)$-vertex graph $H$ with $\eta(H)=O(1)$.
\end{lemma}

A partition $(A,B,C)$ of an $m$-vertex graph $H$ is \emph{balanced} if
$\max(|A|,|B|)\leq 2m/3$, $|C|=O(m^{2/3})$, and $A$ and $B$ are
nonadjacent in $G$.

\begin{lemma}[{Reed and Wood~\cite[Theorem~1.2]{ReedW09}}]
\label{lemma:lemma3}
It takes $O(m)$ time to obtain a balanced partition of an $m$-vertex
graph $H$ with $\eta(H)=O(1)$.
\end{lemma}

\begin{proof}[Proof of Lemma~\ref{lemma:lemma1}] 
Apply Lemma~\ref{lemma:lemma2} to obtain an undirected graph $H$ with
$m=|V(H)|=\Theta({n}/{\log n})$ and $\eta(H)=O(1)$ and an
$H$-partition $(U_1,\ldots,U_m)$ of $G$ in $O(n)$ time.  Within the
proof, we call a member of $V(G)$ a \emph{vertex} of $G$ and call a
member of $V(H)$ a \emph{node} of $H$.  Hence, each node $U_j$ of $H$
with $j\in[m]$ is a set of $O(\log n)$ vertices of~$G$.  For each
vertex $v$ of $G$, let $j_v\in [m]$ denote the index with $v\in
U_{j_v}$.  Let $b=\lceil\log^2 n\rceil$.  Let $T$ be the rooted
ordered binary tree on node subsets of $V(H)$ returned by the
following recursive procedure on $H$:
\begin{itemize}
\item If $|V(H)|\leq b^2$, then return
$V(H)$ as a leaf of $T$. 
\item
If $|V(H)|>b^2$, then obtain a balanced partition $(A,B,C)$ of $H$
in $O(|V(H)|)$ time 
by Lemma~\ref{lemma:lemma3}
and return the
subtree $T(C)$ of $T$ rooted at $C$ whose left (respectively, right)
subtree is the one returned by the procedure call on $H[A\cup C]$
(respectively, $H[B\cup C]$).
\end{itemize}
Let $W_1,\ldots,W_p$ be the leaves of $T$.  Note that each $W_i$ with
$i\in[p]$ is a node subset of $V(H)$.  Let each $V_i$ with $i\in[p]$
consist of the vertices $v\in V(G)$ with $|N_G(v)|\leq b$ such that
$W_i$ is the unique leaf of $T$ with $U_{j_v}\in W_i$.  Let
$V_0=V(G)\setminus (V_1\cup \cdots\cup V_p)$.  Thus, a vertex $v$ of
$G$ belongs to $V_0$ if and only if $|N_G(v)|>b$ or $U_{j_v}\in C$
holds for a nonleaf member $C\in V(T)$.  The rest of the proof shows
that the above procedure runs in $O(n)$ time and $(V_0,\ldots,V_p)$ is
a star partition of $G$.

To show that $(V_0,\ldots,V_p)$ is obtained in $O(n)$ time, let each
$\Psi(C)\subseteq V(H)$ with $C\in V(T)$ denote the union of the
leaves of the subtree $T(C)$ of $T$ rooted at $C$.  Hence,
$\Psi(C)=V(H)$ for the root $C$ of $T$ and $\Psi(C)=C$ for each leaf
$C$ of $T$.  If $X$ and $Y$ are the children of a $C\in V(T)$ in $T$,
then $(\Psi(X)\setminus C,\Psi(Y)\setminus C,C)$ is a balanced
partition of $H[\Psi(C)]$ obtained in $O(|\Psi(C)|)$ time.  Let
$\Lambda_0$ consist of the leaves of $T$.  Let $m_0=\sum_{C\in
  \Lambda_0}|C|$.  For each of the $O(\log m_0)$ levels of $T$, the
balanced partitions of $H[\Psi(C)]$ for all $C\in V(T)$ in the same
level of $T$ take overall $O(m_0)$ time.  Thus, $T$ is obtained in
$O(m_0\log m_0)$ time.  Observe that if $C\in V(T)$, then $C\in
\Lambda_0$ if and only if $|\Psi(C)|\leq b^2$.  Let each $\Lambda_i$
with $i\geq 1$ consist of the $C\in V(T)$ with
\[
b^2\cdot (1.5)^{i-1}< |\Psi(C)|\leq b^2\cdot(1.5)^i.
\] 
If $X$ and $Y$ are distinct members of $\Lambda_i$ such that $X$ is an
ancestor of $Y$ in $T$, then the distance of $X$ and $Y$ in $T$ is
$O(1)$. Thus, \(\sum_{C\in \Lambda_i}|\Psi(C)|=O(m_0)\) holds for each
$i\geq 1$, implying \(|\Lambda_i|=O(\frac{m_0}{b^2}\cdot
(\frac{2}{3})^i)\) by $b^2\cdot (1.5)^{i-1}< |\Psi(C)|$.  Each $C\in
\Lambda_i$ with $i\geq 1$ admits a balanced partition $(A,B,C)$ of
$H[\Psi(C)]$, implying $|C|=O(b^{4/3}\cdot (1.5)^{2i/3})$ by
$|\Psi(C)|\leq b^2\cdot(1.5)^i$.  Since each $C\in V(T)\setminus
\Lambda_0$ contributes $|C|$ to the nonnegative difference $m_0-m$, we
have
\begin{equation}
\label{equation:equation1}
m_0-m
=\sum_{i\geq 1}\sum_{C\in \Lambda_i}|C|
=\sum_{i\geq 1}|\Lambda_i|\cdot O(b^{4/3}\cdot (1.5)^{2i/3})
=O\left(\frac{m_0}{b^{2/3}}\right)\cdot \sum_{i\geq 1} \left(\frac{2}{3}\right)^{i/3}
=O\left(\frac{m_0}{\log^{4/3}n}\right).
\end{equation}
By $m=m_0-o(m_0)$, we have $m_0=\Theta(m)$. Thus, the time of
computing $T$ is $O(m_0\log m_0)=O(m\log m)=O(n)$, implying that
$(V_0,\ldots,V_p)$ can be obtained in $O(n)$ time.

We now show that $(V_0,\ldots,V_p)$ is a star partition of $G$.  To
see Condition~\ref{condition:S1}, let $v\in V_i$ with $i\in[p]$.  By
definition of $T$, the leaf $W_i$ of $T$ containing $U_{j_v}$ also
contains all neighbors of $U_{j_v}$ in $H$.  The union of the nodes of
$H$ in $W_i$ is a subset of $V_0\cup V_i$.  Since $(U_1,\ldots,U_m)$
is an $H$-partition of $G$, we have $N_G(v)\subseteq V_0\cup V_i$.
Thus, $V_1,\ldots,V_p$ are pairwise nonadjacent in $G$.  To see
Condition~\ref{condition:S2}, we have $|V_i|=O(\log^5 n)$ by
$|W_i|\leq b^2$ for each $i\in [p]$ and $|U_j|=O(\log n)$ for each
$j\in[m]$.  By $|N_G(v)|\leq b$ for each $v\in V_i$, we have
$|N_G[V_i]|=O(\log^7 n)=\text{poly}(\log n)$.  As for
Condition~\ref{condition:S3}, $\eta(G)=O(1)$ implies that the number
of vertices $v$ with $|N_G(v)|>b$ is $O(n/b)$.  By
Equation~\eqref{equation:equation1}, we have $\sum_{C\in V(T)\setminus
  \Lambda_0}|C|=O(n/\log^{7/3}n)$.  Thus, $|V_0|=O(n/\log^2
n)+O(n/\log^{4/3} n)=O(n/\log^{4/3} n)$.  By definition of $T$,
$|W_i|=\Omega(b^2)$ holds for each $i\in[p]$. By
$\sum_{i\in[p]}|W_i|=m_0=O(m)$, we have $p=O(m/b^2)=O(n/\log^{4/3}
n)$.  Each $C\in V(T)\setminus \Lambda_0$ contributes $O(|C|\cdot \log
n)$ to the sum $\sum_{i\in[p]}|N_G(V_i)|$. By
Equation~\eqref{equation:equation1}, we have
\[
\sum_{i\in[p]}|N_G(V_i)|=\sum_{C\in V(T)\setminus \Lambda_0}|C|\cdot O(\log
n)=O(m_0/\log^{1/3}n)=O(n/\log^{4/3} n).
\]
Therefore, $|V_0|+p+\sum_{i\in [p]} |N_G(V_i)|=O(n/\log^{4/3} n)$.
\end{proof}

\subsection{The encoding algorithm}
\label{subsection:subsection2.2}

A \emph{concatenation prefix} of $p$ binary strings $X_1,\ldots,X_p$
having overall $n$ bits is an $O(p\log n)$-bit $O(n)$-time computable
string $\chi$ such that it takes $O(1)$ time to obtain from the
concatenation $X$ of $\chi,X_1,\ldots,X_p$ the starting position of
each $X_i$ with $i\in [p]$ in $X$.  For instance, such a $\chi$ can be
the concatenation of strings $\chi_{-1},\chi_0,\chi_1,\ldots,\chi_p$,
each having exactly $b=1+\lceil \log n\rceil$ bits, where (1)
$\chi_{-1}$ is $0^{b-1}1$, (2) $\chi_0$ is the binary representation
of $p$, and (3) each $\chi_i$ with $i\in[p]$ is the binary
representation of the starting position of $X_i$ in the concatenation
of $X_1,\ldots,X_p$.  Thus, it takes $O(1)$ time to obtain $b$ from
the first $O(1)$ words of $X$.  It then takes $O(1)$ time to obtain
$p$ and the starting position of each $X_i$ with $i\in[p]$ in $X$.
The \emph{prefixed concatenation} of binary strings $X_1,\ldots,X_p$
is the concatenation $X$ of $\chi,X_1,\ldots,X_p$ for a concatenation
prefix $\chi$ of $X_1,\ldots,X_p$. When $p=o(n/\log n)$, the prefixed
concatenation of $X_1,\ldots,X_p$ has $n+o(n)$ bits.

The encoding algorithm of our $\cF$-optimal encoding scheme $A^*$ on
an input $n$-vertex graph $G$ has the following three phases: (1) The
first phase computes a decomposition tree $T$ for $G$ via
Lemma~\ref{lemma:lemma1}.  (2) The second phase computes a code-book
string $\chi$ for the subgraphs of $G$ at the leaves of $T$.  (3) The
third phase computes an encoded string $\textit{code}(G)$ for $G$ in a
bottom-up manner along $T$.  The base encoded string
$X_{\textit{base}}$ reported by the encoding algorithm of $A^*$ on $G$
is the prefixed concatenation of $\chi$ and $\textit{code}(G)$.

\paragraph{Phase~1} 
The first phase constructs a height-$O(1)$ decomposition tree $T$ for
$G$ such that (i) each member of $V(T)$, called a \emph{node} of $T$,
is a subgraph of $G$, (ii) the union of all nodes of $T$ is $G$, and
(iii) a node of $T$ is a leaf of $T$ if and only if it has at most
$\ell= \lceil\log \log n\rceil$ vertices.  For each subgraph $H$ of
$G$, let
\[
\partial H=N_G(V(G)\setminus V(H)).
\]
Let the initial tree $T$ consist of the single node $G$.  We
iteratively update $T$ until each leaf node $H$ of $T$ has at most
$\ell$ vertices. The round for a leaf node $H$ of $T$ with
$|V(H)|>\ell$ performs the following:
\begin{enumerate}
\item
If $G$ is equipped with a genus-$O(1)$ embedding that is required to
be recovered together with $G$ by the decoding algorithm, then obtain
a graph $\Delta_H$ from $H$ by adding edges to triangulate each face
of $H$ according to the induced embedding of $H$.  The genus and hence
the Hadwiger number of $\Delta_H$ remain $O(1)$.  Obtain a star
partition $(V_0,\ldots,V_p)$ of $\Delta_H$ in $O(|V(H)|)$ time by
Lemma~\ref{lemma:lemma1}.\footnote{We can only recover a given
genus-$O(1)$ embedding of $G$: $H$ is triangulated into $\Delta_H$ to
ensure the $\cF$-succinctness of the encoded string. The genus of
$\Delta_H$ remains $O(1)$, so $\Delta_H$ belongs to a minor-closed
class of graphs. By $\eta(\Delta_H)=O(1)$, $\Delta_H$ admits an
$O(|V(H)|)$-time obtainable star partition. If the given embedding of
$G$ has genus $\omega(1)$, then we do not know even know whether
$\Delta_H$ admits a star partition.}  Otherwise, obtain a star
partition $(V_0,\ldots,V_p)$ of $H$ in $O(|V(H)|)$ time by
Lemma~\ref{lemma:lemma1}.

\item 
Let $U_0=V_0\cup \partial H$ and $U_i=V_i\setminus \partial H$ for
each $i\in[p]$.  We have $N_H(U_i)\subseteq V_0\cup \partial
H\subseteq U_0$ and $H(U_i)=G(U_i)$ for each $i\in[p]$.  See
Figure~\ref{figure:figure3} for an illustration.  We call
$(U_0,\ldots,U_p)$ the \emph{$T$-partition} of $H$. Let $H_0=H[U_0]$
and $H_i=H(U_i)$ for each $i\in[p]$.  We call $(H_0,\ldots,H_p)$ the
\emph{$T$-subgraphs} of $H$.  Replace the leaf node $H$ of $T$ by a
height-$1$ subtree of $T$ rooted at $H_0$ whose $i$-th child with
$i\in[p]$ is $H_i$.
\end{enumerate}

\begin{figure}
\centerline{\scalebox{0.4}{\includegraphics{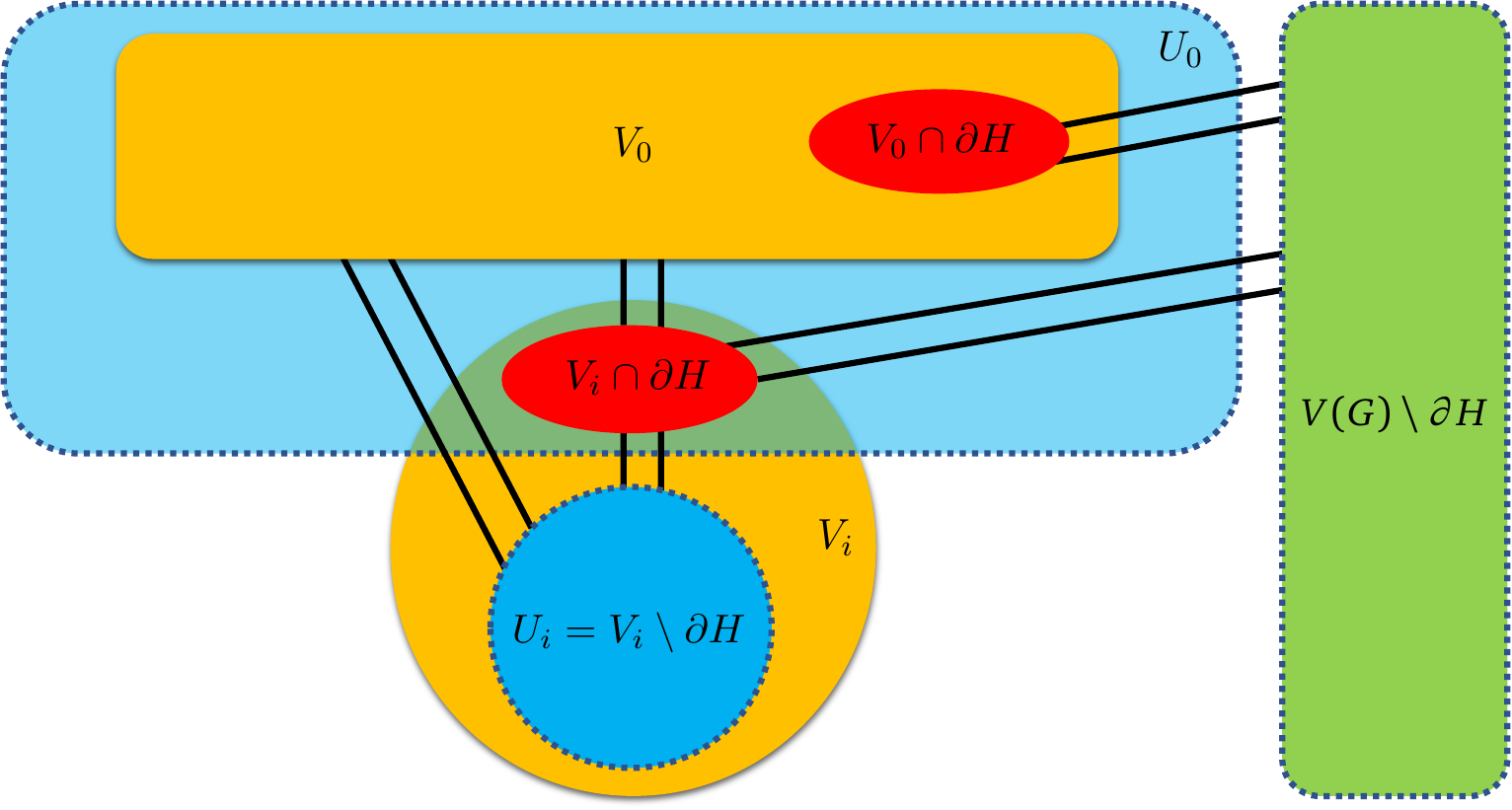}}}
\caption{An illustration for $U_0$ and a $U_i$ with $i\in[p]$.}
\label{figure:figure3}
\end{figure}

Each edge of $G$ belongs to exactly one node of $T$, but a vertex of
$G$ may belong to more than one nodes of $T$.  By
Conditions~\ref{condition:S1} and~\ref{condition:S3} of
$(V_0,\ldots,V_p)$, each of the above rounds increases the overall
number of vertices in all nodes of $T$ by
\begin{equation}
\label{equation:equation2}
\sum_{i\in [p]}|N_H(U_i)|\leq \sum_{i\in[p]}
|N_H(V_i)|+|V_i\cap\partial H|\leq |\partial
H|+O(|V(H)|/\log^{4/3}|V(H)|).
\end{equation}
By Condition~\ref{condition:S2} of $(V_0,\ldots,V_p)$, the height of
$T$ is $O(1)$.  By $\partial H_i\subseteq N_H(U_i)$ for each $i\in[p]$
and $\partial G=\varnothing$, the overall number of vertices in all
nodes of $T$ is $n+o(n)$.  Phase~1 runs in $O(n)$ time.

A subtree of $T$ is denoted $T_H$ for a subgraph $H$ of $G$ if $H$ is
the union of the nodes of $T_H$. Thus, $T_H=H$ for each leaf node $H$
of $T$ and $T_G=T$.  Note that distinct subtrees $T_H$ and $T_{H'}$ of
$T$ may have isomorphic subgraphs $H$ and $H'$ of $G$, but their
$|\partial H|$ and $|\partial H'|$ can still be different.

\paragraph{Phase~2}
For each positive integer $k\leq \ell$, let $\Lambda_k$ consist of the
leaf nodes $H$ of $T$ with $|V(H)|\leq k$.  Let each $\Lambda^*_k$
with $k\in[\ell]$ consist of each graph $H$ in $\Lambda_k$ such that
at least one of its occurrences in $G$ identified by the tree
structure of $T$ admits a vertex subset $U\subseteq V(H)$ with
$H=G(U)$ and $|N_G(U)| \leq |V(H)|/\log^2 |V(H)|$.  Thus, each graph
in $\Lambda^*_k$ is an $O(k/\log^2 k)$-quasi-member of each slim class
$\cC$ that contains $G$.  By $O(1)^{\textit{poly}(\ell)}=o(n)$, it
takes $O(n)$ time to
\begin{itemize}
\item 
design for each distinct $k$-vertex graph $H$ in $\Lambda_\ell$ a
unique encoded string $\textit{code}(H)$ having at most $\log
|\Lambda_k|+O(1)$ bits such that if $H$ is in $\Lambda^*_\ell$, then
$\textit{code}(H)$ has at most $\log |\Lambda^*_k|+O(1)$ bits and

\item
construct an $o(n)$-bit code-book string $\chi$ with which each graph
$H\in \Lambda_\ell$ and $\textit{code}(H)$ can be obtained from each
other in $O(1)$ time.
\end{itemize}
Note that whether two leaf nodes of $T$ are considered distinct
elements of $\Lambda_\ell$ depends on whether they are isomorphic and
equipped with the same additional information to be recovered by the
decoding algorithm.  Two distinct subtree $T_H$ and $T_{H'}$ with
isomorphic leaf nodes $H$ and $H'$ of $T$ may have different
$|\partial H|$ and $|\partial H'|$, though.  The higher complexity the
additional information has, the longer the code-book string $\chi$ and
the encoded strings $\textit{code}(H)$ for the tiny subgraphs $H$ of
$G$ in $\Lambda_\ell$ become.  The base encoded string
$X_{\textit{base}}$ for the input graph $G$ is the prefixed
concatenation of the $o(n)$-bit code-book string $\chi$ and the
$\cF$-succinct encoded string $\textit{code}(G)$ for $G$ to be
computed in the third phase.

\paragraph{Phase~3}
The encoded string $\textit{code}(G)$ is defined recursively for each
subtree $T_H$ of $G$: If $|V(H)|\leq \ell$, then let
$\textit{code}(H)$ be as defined in the code-book string $\chi$ such
that $\textit{code}(H)$ and $H$ can be obtained from each other in
$O(1)$ time via $\chi$. Otherwise, let $(U_0,\ldots,U_p)$ and
$(H_0,\ldots,H_p)$ be the $T$-partition and $T$-subgraphs of $H$.  We
already have each $\textit{code}(H_i)$ with $i\in[p]$, since $T_{H_i}$
is the $i$-th subtree of $T_H$.  The encoded string $\textit{code}(H)$
for $H$ is the prefixed concatenation of
$\textit{code}(H_0),\ldots,\textit{code}(H_p)$, where
$\textit{code}(H_0)$ is the following $O(|V(H)|)$-time computable
$O(|U_0|\cdot \log |V(H)|)+o(|V(H)|)$-bit string with which the
decoding algorithm can recover $H$ and its equipped information from
the subgraphs $H_i$ with $i\in[p]$ and their equipped information in
$O(|V(H)|)$ time.
\begin{itemize}
\item 
Each edge of $H_0$ and each duplicated copy of each vertex of $U_0$ in
$H_i$ with $i\in[p]$ can be represented using $O(\log|V(H)|)$ bits. By
$\eta(H)=O(1)$, we have $|E(H_0)|=O(|U_0|)=O(|V_0|+|\partial H|)$.  By
Equation~\eqref{equation:equation2}, we have
\[
\sum_{i\in[p]}|N_H(U_i)|\leq |U_0|+O(|V(H)|/\log^{4/3}|V(H)|).
\] 
Thus, $H$ together with its equipped information like vertex and edge
coloring and orientation can be recovered from $H_1,\ldots,H_p$ and
their equipped information using an $O(|U_0|\cdot \log
|V(H)|)+o(|V(H)|)$-bit string $\textit{code}(H_0)$.

\item
For the case that $G$ is equipped with a genus-$O(1)$ embedding
required to be recovered together with $G$ by the decoding algorithm,
consider the clockwise order of the incident edges of each vertex $u$
in $H$ around $u$ according to the induced embedding of $H$.
\begin{itemize}
\item
If $u\in U_0$, then the induced embedding of each $H_i$ with $i\in[p]$
preserves the induced order of the incident edges of $u$ in $H_i$
around $u$.
\item
If $u\in U_i$ with $i\in [p]$, then $H_i$ contains all incident edges
of $u$, implying that the induced embedding of $H_i$ preserves their
order around $u$.
\end{itemize}
Hence, $\textit{code}(H_0)$ uses $O(\log |V(H)|)$ bits to encode
$H[\{u,v,w\}]$ for every triple of vertices $u\in U_0$, $v\in U_j$,
and $w\in U_k$ with $jk=0$ or $j\ne k$ such that $\vec{uv}$ or
$\vec{vu}$ immediately precedes $\vec{uw}$ or $\vec{wu}$ around $u$ in
clockwise order around $u$ according to the induced embedding of $H$.
\begin{itemize}
\item 
The number of such triples $(u,v,w)$ with $jk=0$ is $O(|U_0|)$, since
each of the $O(|U_0|)$ edges of $H_0$ belongs to $\Theta(1)$ such
subgraphs $H_0[\{u,v,w\}]$.
\item
The number of such triples $(u,v,w)$ with $jk\ne 0$ and $j\ne k$ is
also $O(|U_0|)$: Vertices sets $U_j\subseteq V_j$ and $U_k\subseteq
V_k$ with $jk\ne 0$ and $j\ne k$ are non-adjacent in the triangulated
version $\Delta_H$ of $H$ by Condition~\ref{condition:S1} of
$(V_0,\ldots,V_p)$, implying that one of the $O(|U_0|)$ incident edges
of $u$ in $\Delta_H[U_0]$ succeeds $\vec{uv}$ or $\vec{vu}$ and
precedes $\vec{uw}$ or $\vec{wu}$ in the clockwise order around $u$
according to the embedding of $\Delta_H$.
\end{itemize}
\end{itemize}
Thus, $\textit{code}(H_0)$ has $O(|U_0|\cdot\log |V(H)|)+o(|V(H)|)$
bits.  Since the height of $T$ is $O(1)$ and the overall number of
vertices in all nodes of $T$ is $n+o(n)$, the third phase also runs in
$O(n)$ time.

\paragraph{Encoding size}
A function $f$ is \emph{super-additive} and
\emph{continuous}~\cite{HeKL00} if $f(n_1)+f(n_2) \leq f(n_1 +n_2)$
and $f(n + o(n)) = f(n) + o(f(n))$, respectively.  For example, $f(n)
= n^a \log ^b n$ for any constants $a \geq 1$ and $b \geq 0$ is
continuous and super-additive.  Super-additivity and continuity are
both closed under additions.  By
$\|X_{\textit{base}}\|=\|\textit{code}(G)\|+o(n)$ and $\log
|\cC_n|=\Theta(n)$ for each class $\cC\in \cF$, we ensure the
$\cF$-succinctness of $X_{\textit{base}}$ by proving
\begin{equation}
\label{equation:equation3}
\|\textit{code}(G)\|\leq f(n)+o(f(n))
\end{equation}
for each continuous super-additive function $f$ and each slim class
$\cC$ containing the input $n$-vertex graph $G$ that satisfy
$\log|\cC_n|\leq f(n)+o(f(n))$.

We first ready
Equations~\eqref{equation:equation4},~\eqref{equation:equation5},
and~\eqref{equation:equation6} below that are needed to prove
Equation~\eqref{equation:equation3}.  For each subtree $T_H$ of $T$
with $k=|V(H)|>\ell$, the encoded string $\textit{code}(H)$ is the
prefixed concatenation of
$\textit{code}(H_0),\ldots,\textit{code}(H_p)$ for the $T$-subgraphs
$(H_0,\ldots,H_p)$ of $H$.  For the $T$-partition $(U_0,\ldots,U_p)$
of $H$ and its corresponding star partition $(V_0,\ldots,V_p)$, we
have $U_0=V_0\cup \partial H$ and for each $i\in[p]$
\begin{align*}
V(H_i)&=U_i\cup N_H(U_i)\\
\partial H_i\cup N_H(U_i)&\subseteq N_H(V_i)\cup (V_i\cap \partial H).
\end{align*}
By $\|\textit{code}(H_0)\|=O(|U_0|\cdot\log k)+o(k)$ and
Conditions~\ref{condition:S1} and~\ref{condition:S3} of
$(V_0,\ldots,V_p)$, we have
\begin{align}
\|\textit{code}(H_0)\|
&=O(|\partial H|\cdot \log k)+o(k)\nonumber\\
|\partial H_1|+\cdots+|\partial H_p|
&\leq |\partial H|+o(k/\log k)\label{equation:equation4}\\
|V(H_1)|+\cdots+|V(H_p)|&\leq |\partial H|+k+o(k/\log k).\nonumber
\end{align}
We can then prove for each subtree $T_H$ of $T$, no matter whether
$k=|V(H)|$ is more than $\ell$ or not,
\begin{equation}
\label{equation:equation5}
\|\textit{code}(H)\|=O(k+|\partial H|\cdot \log k)
\end{equation}
by induction on the bounded height of $T_H$: If $k\leq \ell$, then
$\|\textit{code}(H)\|\leq\log |\Lambda_k|+O(1)=O(k)$ by
$\eta(H)=O(1)$.  If $k>\ell$, then Equation~\eqref{equation:equation4}
and $p=o(k/\log k)$ imply
\[
\|\textit{code}(H)\|
=O(|\partial H|\cdot\log k)+o(k)+\sum_{i\in[p]}O(|V(H_i)|+|\partial H_i|\cdot\log |V(H_i)|)
=O(k+|\partial H|\cdot\log k)
\]
by the inductive hypothesis. Equation~\eqref{equation:equation5} is
proved.  Each $k$-vertex graph $H$ in $\Lambda^*_k$ can be represented
by an encoding of a $k+o(k/\log k)$-vertex graph $H'$ in $\cC$ and an
$o(k)$-bit string specifying a set $U'\subseteq V(H')$ of $o(k/\log
k)$ vertices with $H=H'-U'$.  We have
\begin{align}
\log |\Lambda^*_k|&\leq \log |\cC_{k+o(k)}|+o(k)\nonumber\\
&\leq f(k+o(k))+o(f(k+o(k)))+o(k)\nonumber\\
&\leq f(k)+o(f(k)),\label{equation:equation6}
\end{align}
since $f$ is continuous and super-additive.

We now prove Equation~\eqref{equation:equation3} using
Equations~\eqref{equation:equation4},~\eqref{equation:equation5},
and~\eqref{equation:equation6}.  Note that we have $\partial
G=\varnothing$, $T_G=T$, $G=G(V(G))$, and $|N_G(V(G))|=0$.  Thus, it
suffices to show for each subtree $T_H$ of $T$ with $k=|V(H)|$ that if
$H=G(U)$ holds for a vertex set $U\subseteq V(H)$ with $|N_G(U)|\leq
k/\log^2 k$, then
\begin{equation}
\label{equation:equation7}
\|\textit{code}(H)\|\leq f(k)+o(f(k)).
\end{equation}
We prove Equation~\eqref{equation:equation7} by induction on the
bounded height of $T_H$.  If $k\leq \ell$, then $H=G(U)$ belongs to
$\Lambda^*_k$ by $|N_G(U)|\leq k/\log^2 k$, implying
\[
\|\textit{code}(H)\|\leq\log |\Lambda^*_k|+O(1) \leq f(k)+o(f(k)).
\]
The basis holds. If $k>\ell$, then $H=G(U)$ implies $\partial
H\subseteq N_G(U)$.  By $|N_G(U)|\leq k/\log^2 k$, we have
\begin{equation}
\label{equation:equation8}
|\partial H|\leq k/\log^2 k.
\end{equation}
Let $(U_0,\ldots,U_p)$ and $(H_0,\ldots,H_p)$ be the $T$-partition and
the $T$-subgraphs of $H$.  By Condition~\ref{condition:S2} of the
corresponding star partition $(V_0,\ldots,V_p)$ of $H$, we have
\[
|V(H_i)|=|N_H[U_i]|\leq |N_H[V_i]|=\textit{poly}(\log k)
\]
for each $i\in[p]$.  Let $I$ consist of the indices $i\in[p]$ such
that $H_i$ admits subset $U'\subseteq V(H_i)$ with $H_i=G(U')$ and
$|N_G(U')|\leq |V(H_i)|/\log^2 |V(H_i)|$.  If $i\notin I$, then
$H_i=H(U_i)=G(U_i)$ implies
\begin{equation}
\label{equation:equation9}
|N_H(U_i)|=|N_G(U_i)|>
\frac{|V(H_i)|}{\log^2 |V(H_i)|}
=\Omega\left(\frac{|V(H_i)|}{\log^2 \log k}\right).
\end{equation}
By Equations~\eqref{equation:equation2},~\eqref{equation:equation8},
and~\eqref{equation:equation9}, we have
\begin{align}
\sum_{i\in [p]\setminus I} |V(H_i)|&=
\sum_{i\in [p]\setminus I} O(|N_H(U_i)|\cdot \log^2 \log k)\nonumber\\
&=O((|\partial H|+k/\log^{4/3}k)\cdot \log^2\log k)\nonumber\\
&=o(k/\log k).\label{equation:equation10}
\end{align}
By $f(k)=\Omega(k)$, the inductive hypothesis, and 
Equations~\eqref{equation:equation4}, \eqref{equation:equation5},
\eqref{equation:equation8},
and~\eqref{equation:equation10}, we have
\begin{align*}
\|\textit{code}(H)\|&=O(|\partial H|\cdot \log k)+o(k)+
\sum_{i\in [p]\setminus I}\|\textit{code}(H_i)\|+
\sum_{i\in I}\|\textit{code}(H_i)\|\\
&\leq o(k)+\sum_{i\in [p]\setminus I} O(|V(H_i)|+|\partial H_i|\cdot\log |V(H_i)|)+
\sum_{i\in I}f(|V(H_i)|)+o(f(|V(H_i)|))\\
&\leq o(k)+\sum_{i\in [p]} O(|\partial H_i|\cdot\log k)+
\sum_{i\in [p]}f(|V(H_i)|)+o(f(|V(H_i)|))\\
&\leq f(k)+o(f(k)).
\end{align*}
Equation~\eqref{equation:equation7} is proved, implying that
Equation~\eqref{equation:equation3} holds.  Since $X_{\textit{base}}$
is $\cC$-succinct for each slim class $\cC$ that contains $G$,
$X_{\textit{base}}$ is $\cF$-succinct. Thus, the encoded string $X$
produced by the $O(n)$-time encoding algorithm of $A^*$ for the input
$n$-vertex graph $G$ is $\cF$-succinct.

\subsection{The decoding algorithm}
\label{subsection:subsection2.3}
The decoding algorithm of $A^*$ first obtains the $o(n)$-bit code-book
string $\chi$ and the binary string $\textit{code}(G)$ for $G$ from
$X_{\textit{base}}$.  It takes $O(n)$ time to recover the
height-$O(1)$ decomposition tree $T$ of $G$ obtained by the encoding
algorithm and the string $\textit{code}(H)$ for each subtree $T_H$ of
$T$. It takes overall $O(n)$ time to obtain all subgraphs $H$ of $G$
at the leaves of $T$ from $\chi$ and $\textit{code}(H)$.  For each
non-singleton subtree $T_H$ of $T$, it takes $O(|V(H)|)$ time to
recover $H$ and its equipped additional information, if any, from
$\textit{code}(H_0)$ and the subgraphs $H(U_i)$ with $i\in[p]$ with
respect to the $T$-partition $(U_0,\ldots,U_p)$ of $H$ used by the
encoding algorithm of $A^*$ which is preserved in the tree structure
of $T_H$.  Thus, it takes $O(n)$ time to decode $X$ back to the graph
$G$ and its equipped additional information.

For the case that the embedding reflected by the input adjacency list
of $G$ need not be recovered together with $G$, the decoding algorithm
of $A^*$ reports an adjacency list of $G$ that is recursively defined
for each subtree $T_H$ of $T$ as follows.  For each distinct subgraph
$H$ of $G$ in $\Lambda_\ell$, we report the adjacency list of $H$
stored in the code-book string $\chi$. Note that this need not be the
same as the induced adjacency list of any occurrence of $H$ in $G$.
For each non-singleton subtree $T_H$ of $T$, let $(U_0,\ldots,U_p)$
and $(H_0,\ldots,H_p)$ be the $T$-partition and the $T$-subgraphs of
$H$.  For each vertex $u\in U_i$ with $i\in[p]$, the neighbor list of
$u$ in $H$ is exactly the reported neighbor list of $u$ in $H_i$.  For
each $u\in U_0$, the neighbor list of $u$ in $H$ is the concatenation
of the reported neighbor lists of $u$ in $H_0,\ldots,H_p$ in order.

\subsection{Three non-monotone slim classes of graphs}
\label{subsection:subsection2.4}
To ensure that our $\cF$-optimal encoding scheme $A^*$ indeed
overshadows all the previous work for non-monotone classes of graphs
listed in Figure~\ref{figure:figure1}, we show that the three
non-monotone classes $\cC$ of triconnected planar graphs and
triangulations and floor-plans of genus-$O(1)$ surfaces are
quasi-monotone.  Specifically, we verify that the subgraph $G(U)$ of
$G$ for each nonempty proper subset $U$ of $V(G)$ can be obtained from
a graph of $\cC$ by deleting $O(|N_G(U)|)$ vertices and their incident
edges.  Since all graphs in these three classes $\cC$ are connected,
we have $|N_G(U)|\geq 1$.
\begin{itemize}
\item
Example~1: triconnected planar graphs.  Let $H=G(U)$. Repeat the
following three steps on the current graph $H$ for three iterations:
(i) Let the current graph $H$ be embedded such that each of its
biconnected components is incident to the exterior face.  (ii) Add a
new vertex into the exterior face of $H$.  (iii) Make the new vertex a
common neighbor of all vertices on the exterior bound of the current
$H$.  The plane graph $H$ at the end of the $k$-th round is
$k$-connected.  Thus, the initial $H$ can be obtained from the final
$H\in \cC$ by deleting exactly three vertices and their incident
edges.

\item
Example~2: triangulations of a genus-$O(1)$ surface.  The boundary of
each non-triangle face $F$ of $G(U)$ contains at least two vertices of
$N_G(U)$.  Let $e_F$ be an edge between arbitrary two vertices of
$N_G(U)$ on the boundary of $F$.  The graph consisting of the edges
$e_F$ for all non-triangle faces $F$ of $G(U)$ has genus $O(1)$,
implying that the number of non-triangle faces of $G(U)$ is
$O(|N_G(U)|)$.  Thus, adding a new vertex $u_F$ to triangulate each
non-triangle face $F$ of $G(U)$ results in a triangulation of a
genus-$O(1)$ surface.  $G(U)$ can be obtained from the resulting graph
by deleting the $O(|N_G(U)|)$ new vertices $u_F$ and their incident
edges.

\item 
Example~3: floor-plans of a genus-$O(1)$ surface.  Since each vertex a
floor-plan has $O(1)$ degree, one can make $G(U)$ a floor-plan by
adding $O(|N_G(U)|)$ vertices and edges such that $G(U)$ can be
obtained from the resulting floor-plan by deleting the new vertices
and their incident edges.
\end{itemize}
Note that there is no need to modify our $A^*$ to accommodate these
non-monotone classes $\cC$ as did by the previous encoding schemes in
the literature. Since the above classes are clearly nontrivial and
their members have bounded Hadwiger numbers, we just have to prove
that they are quasi-monotone and then our $\cF$-optimal encoding
scheme $A^*$ is guaranteed to be $\cC$-optimal encoding schemes for
each of the above classes $\cC$.  One can obtain more examples this
way.

\section{The query algorithms}
\label{section:section3}
Section~\ref{subsection:subsection3.1} presents a framework for
designing an $O(n)$-time obtainable $o(n)$-bit string $X_{q}$ for a
query $q$ such that an answer to $q$ can be obtained in $O(1)$ time
from $X_{\textit{base}}$ and $X_q$.
Section~\ref{subsection:subsection3.2} applies the framework on the
query of obtaining the degree of a vertex $u$ in $G$.  Other queries
for a vertex of $G$ can be supported in a same way.
Section~\ref{subsection:subsection3.3} applies the framework on the
query of determining whether $\vec{uv}$ is an edge of $G$ for a pair
$(u,v)$ of vertices of $G$.  Other queries for a pair of adjacent
vertices of $G$ can be supported in a same way.
Section~\ref{subsection:subsection3.4} applies the framework on the
query $q_t$ of reporting a shortest $uv$-path of $G$ for a pair
$(u,v)$ of vertices $u$ and $v$ of $G$ if there is one whose length is
bounded by a prespecified $t=O(1)$.

\subsection{A framework for supporting $\boldsymbol{O(1)}$-time queries using additional $\boldsymbol{o(n)}$ bits}
\label{subsection:subsection3.1}
This subsection presents a framework for supporting a query $q$ in
$O(1)$ time using an $O(n)$-time obtainable $o(n)$-bit encoded string
$X_{q}$ which is the prefixed concatenation of (i) a string
$\chi_{\textit{label}}$ for labeling, (ii) a string
$\chi_{\textit{leaf}}$ supporting the query for the tiny graphs in
$\Lambda_\ell$, and (iii) a string $\chi_G$ to be recursively defined
based on the decomposition tree $T$ for $G$.

\paragraph{Dictionary}
Our framework uses Lemma~\ref{lemma:lemma4} below to handle the labels
of vertices.  Let each $Y[i]$ with $i\in[m]$ denote the $i$-th bit of
an $m$-bit binary string $Y$.  For each $i\in[m]$, let
$\textit{rank}(Y,i)=Y[1]+\cdots+Y[i]$.  Let $\textit{rank}(Y)$ denote
the number $\textit{rank}(Y,\|Y\|)$ of $1$-bits in $Y$.  Let each
$\textit{select}(Y,j)$ with $j\in[\textit{rank}(Y)]$ denote the index
$i\in[m]$ such that $Y[i]$ is the $j$-th $1$-bit in $Y$.  A
\emph{fully indexable dictionary}~\cite{RamanRS07} for $Y$ is a binary
string from which (1) each $\textit{select}(Y,j)$ with
$j\in[\textit{rank}(Y)]$, (2) each $\textit{rank}(Y,i)$ with
$i\in[m]$, and (3) each $Y[i]$ with $i\in [m]$ can be obtained in
$O(1)$ time.

\begin{lemma}
\label{lemma:lemma4}
It takes $O(m)$ time to compute an $O(r\log m)+o(m)$-bit fully
indexable dictionary $\textit{dict}(Y)$ for an $m$-bit binary string
$Y$ with $\textit{rank}(Y)=r$.
\end{lemma}

\begin{proof}
Let $h=\lceil\frac{1}{2}\log m\rceil$.  Assume $r\geq 1$ and that $m$
is an integral multiple of $h^2$ without loss of generality.  Let
$\textit{dict}(Y)$ be the prefixed concatenation of the following
$O(m)$-time obtainable $O(r\log m)+o(m)$-bit strings
$\textit{dict}_1(Y)$, $\textit{dict}_2(Y)$, and $\textit{dict}_3(Y)$.

(1) Select: The select query can be supported in $O(1)$ time by the
string $\textit{dict}_1(Y)$ whose $j$-th $2h$-bit word for each $j\in
[r]$ stores $\textit{select}(Y,j)$. We have
$\|\textit{dict}_1(Y)\|=O(r\log m)$.

(2) Rank: The rank query can be supported in $O(1)$ time by the
prefixed concatenation $\textit{dict}_2(Y)$ of the following
$O(m)$-time obtainable $o(m)$-bit strings
$\chi_{2a},\chi_{2b},\chi_{2c}$:
\begin{itemize}
\item 
For each $i\in [m/h^2]$, let the $i$-th $2h$-bit word $\chi_{2a}(i)$
of $\chi_{2a}$ store $\textit{rank}(Y,ih^2)$. Thus, $\|\chi_{2a}\|=o(m)$.
\item 
For each $i\in[m/h^2]$, let $Y_i$ be the $i$-th $h^2$-bit substring
$Y[(i-1)h^2+1,ih^2]$ of $Y$ and let $\chi_{2b}$ store the ranks for
the positions that are integral multiples of $h$ in all $Y_i$ with
$i\in[m/h^2]$.  Specifically, let $\chi_{2b}$ be the concatenation of
all $\chi_{2b}(i)$ with $i\in[m/h^2]$, where the $j$-th $2\lceil \log
h\rceil$-bit word of $\chi_{2b}(i)$ stores $\textit{rank}(Y_i,jh)$ for
each $j\in[h]$.  Thus, $\|\chi_{2b}\|=O((m/h^2)\cdot h\cdot \log
h)=o(m)$.

\item
By $2^{h}=O(\sqrt{m})$, there is an $O(m)$-time computable $o(m)$-bit
string $\chi_{2c}$ answering $\textit{rank}(Z)$ in $O(1)$ time for any
string $Z$ having at most $h$ bits.
\end{itemize}
For each $k\in[m]$, it takes $O(1)$ time to obtain $\textit{rank}(Y,
k)= r_a+r_b+r_c$ as follows.  With $i=\lceil k/h^2\rceil$, obtain
$r_a=\textit{rank}(Y,(i-1)h^2)$ from $\chi_{2a}$.  With
$j=\lceil(k-(i-1)h^2)/h\rceil$, obtain $r_b=\textit{rank}(Y_i,(j-1)h)$
from $\chi_{2b}$.  Obtain $r_c=\textit{rank}(Y[(i-1)h^2+(j-1)h+1,k])$
from $\chi_{2c}$.

(3) Membership: The membership query can be supported in $O(1)$ time
by the prefixed concatenation $\textit{dict}_3(Y)$ of the following
strings $\chi_{3a}$, $\chi_{3b}$, and $\chi_{3c}$:
\begin{itemize}
\item 
Let $\chi_{3a}$ be the $m/h$-bit string such that $\chi_{3a}[i]=0$
if and only if the $i$-th $h$-bit word of $Y$ is $0$.

\item 
Let $\chi_{3b}=\textit{dict}_2(\chi_{3a})$, which has $o(m)$ bits.

\item 
Let $\chi_{3c}$ be the $hr$-bit string such that if $\chi_{3a}[i]$ is
the $j$-th $1$-bit of $\chi_{3a}$, then the $j$-th $h$-bit word of
$\chi_{3c}$ with $j\in [r]$ stores the $i$-th $h$-bit word of $Y$.
\end{itemize}
For each $k\in[m]$, it takes $O(1)$ time to obtain $Y[k]$ as follows.
Let $i=\lceil k/h\rceil$.  If $\chi_{3a}[i]=0$, then $Y[k]=0$. If
$\chi_{3a}[i]=1$, then obtain $j=\textit{rank}(\chi_{3a},i)$ from
$\chi_{3b}$ and obtain $Z=Y[(i-1)h+1,ih]$ from the $j$-th $h$-bit word
of $\chi_{3c}$.  Thus, $Y[k]=Z[k-(i-1)h]$.
\end{proof}

\paragraph{Labeling}
A \emph{labeling} for a graph $H$ is a bijection $L:V(H)\to [|V(H)|]$,
assigning each vertex $u$ of $H$ a \emph{label} $L(u)$ representable
in $\lceil \log|V(H)|\rceil$ bits.  The query support of our encoding
scheme $A^*$ is based on the following $O(n)$-time obtainable labeling
$L=L_G$.  By $O(1)^{\textit{poly}(\ell)}=o(n)$, it takes $O(n)$ time
to associate an arbitrary fixed labeling $L_H$ to each distinct
subgraph $H$ in $\Lambda_\ell$ and construct an $o(n)$-bit string
$\chi_{\textit{leaf}}$ from which $L_H$ can be obtained in $O(1)$ time
via $\textit{code}(H)$.  For each subtree $T_H$ of $T$ with
$k=|V(H)|>\ell$, let let $(U_0,\ldots,U_p)$ and $(H_0,\ldots,H_p)$ be
the $T$-partition and $T$-subgraphs of $H$.  Let each $L_H(u)$ with
$u\in U_0$ be an arbitrary distinct integer in $[|U_0|]$.  For each
$u\in U_i$ with $i\in[p]$, if $L_{H_i}(u)$ is the $j$-th smallest
number in the set $L_{H_i}(U_i)$ of labels, then let
$L_H(u)=|U_0|+\cdots+|U_{i-1}|+j$.

Let $\chi_{\textit{label}}$ be the prefixed concatenation of
$\chi_{\textit{leaf}}$ and $\chi_G$, where $\chi_H$ for each subtree
$T_H$ of $T$ is recursively defined as follows.  Let $\chi_H$ be an
$O(1)$-bit string fixed for all subgraphs $H\in \Lambda_\ell$,
signifying that $H$ is a leaf node of $T$ and $\textit{code}(H)$ can
be obtained from $X_{\textit{base}}$ in $O(1)$ time using the position
of $T_H$ in $T$.  If $k=|V(H)|>\ell$, then let $\chi_H$ be the
prefixed concatenation of $\chi_{H_0},\ldots,\chi_{H_p}$ for the
$T$-partition $(U_0,\ldots,U_p)$ and the $T$-subgraphs
$(H_0,\ldots,H_p)$ of $H$ such that $\chi_{H_0}$ supports the
following queries in $O(1)$ time:
\begin{enumerate}[label=\emph{L\arabic*:}, ref={L\arabic*}]
\item 
\label{condition:L1}
Given $L_H(u)$, obtain the index $i\in[0,p]$ with $u\in U_i$.

\item 
\label{condition:L2}
Given $i\in[0,p]$ and $L_{H_i}(u)$, obtain $L_H(u)$.
\end{enumerate}

(1) For Query~\ref{condition:L1}, let $Y_0$ be the $k$-bit string
whose $j$-th $1$-bit with $j\in [p]$ is at position
$|U_0|+\cdots+|U_{j-1}|+1$.  The index $i\in[0,p]$ with $u\in U_i$ is
$\textit{rank}(Y_0,L_H(u))$, which is obtainable in $O(1)$ time from
the $O(p\log k)+o(k)$-bit string $\chi_0=\textit{dict}(Y_0)$ by
Lemma~\ref{lemma:lemma4}. By Condition~\ref{condition:S3} of the star
partition $(V_0,\ldots,V_p)$ corresponding to $(U_0,\ldots,U_p)$, we
have $\|\chi_0\|=o(k)$.

(2) For Query~\ref{condition:L2}, we focus on the case with $i\in[p]$,
since $L_H(u)=L_{H_0}(u)$.  For each $i\in [p]$, let
$\chi_{i,a}=\textit{dict}(Y_i)$ for the $|V(H_i)|$-bit string $Y_i$
such that $Y_i[j]=1$ if and only if $U_0$ contains the vertex $u$ with
$L_{H_i}(u)=j$.  Lemma~\ref{lemma:lemma4} implies
$\|\chi_{i,a}\|=O(|N_H(U_i)|\cdot\log k)+o(|V(H_i)|)$.  Let
$\chi_{i,b}$ be the $O(|N_H(U_i)|\cdot \log k)$-bit string whose
$j$-th $\lceil\log k\rceil$-bit word $\chi_{i,b}(j)$ with $j\in
[|N_H(U_i)|]$ stores $L_H(u)$ for the vertex $u$ of $H_i$ such that
$Y_i[L_{H_i}(u)]$ is the $j$-th $1$-bit of $Y_i$ (i.\,e.,
$\textit{select}(Y_i,j)=L_{H_i}(u)$).  To obtain $L_H(u)$ from $i\in
[p]$ and $L_{H_i}(u)$, obtain $c=Y_i[L_{H_i}(u)]$ from $\chi_{i,a}$.
If $c=0$, then obtain $L_H(u)=L_{H_i}(u)+\textit{select}(Y_0,i)-1$
from $\chi_0$.  If $c=1$, then obtain
$L_H(u)=\chi_{i,b}(\textit{rank}(Y_i,L_{H_i}(u)))$ from $\chi_{i,a}$
and $\chi_{i,b}$.  Let $\chi_i$ be the $O(|N_H(U_i)|\cdot\log
k)+o(|V(H_i)|)$-bit prefixed concatenation of $\chi_{i,a}$ and
$\chi_{i,b}$.  The bit count of the $O(k)$-time obtainable prefixed
concatenation $\chi_{H_0}$ of the strings $\chi_0,\ldots,\chi_p$ is
\begin{equation}
\label{equation:equation11}
o(k)+\sum_{i\in[p]} O(|N_H(U_i)|\cdot\log
k)+o(|V(H_i)|)=o(k)+O(|\partial H|\cdot \log k).
\end{equation}
We prove the next equation for each subtree $T_H$ of $T$ by induction
on the bounded height of $T_H$:
\begin{equation}
\label{equation:equation12}
\|\chi_H\|=o(k)+O(1+|\partial H|\cdot \log k).
\end{equation}
If $k\leq \ell$, then $\|\chi_H\|=O(1)$ implies
Equation~\eqref{equation:equation12} (even if $k=O(1)$ and $|\partial
H|=0$).  The basis holds.  If $k>\ell$, then
Equations~\eqref{equation:equation4} and~\eqref{equation:equation11}
and $p=o(k/\log k)$ (by Condition~\ref{condition:S3} of
$(V_0,\ldots,V_p)$) imply
\begin{align*}
\|\chi_H\|
&=o(k)+O(|\partial H|\cdot\log k)+\sum_{i\in[p]}o(|V(H_i)|)+O(1+|\partial H_i|\cdot\log |V(H_i)|)\\
&=o(k)+O(1+|\partial H|\cdot\log k)
\end{align*}
by the inductive hypothesis. Equation~\eqref{equation:equation12} is
proved.  By $|\partial G|=0$, we have $\|\chi_G\|=o(n)$.

\paragraph{The framework and how to use it}
Since each graph $H$ in $\Lambda_\ell$ has at most $\ell=O(\log\log
n)$ vertices, it takes $O(n)$ time to compute an $o(n)$-bit string
$\chi_{q,\textit{leaf}}$ from which the query $q$ can be supported in
$O(1)$ time via the labels $L_H(u)$ and other information of the
vertices $u$ given to the query algorithm for $q$ and the encoded
string $\textit{code}(H)$ for $H$ which is $O(1)$-time obtainable from
$X_{\textit{base}}$ using the position of $T_H$ in $T$.  For each
subtree $T_H$ of $T$, recursively define $\chi_{q,H}$ as follows.  Let
$\chi_{q,H}$ be an $O(1)$-bit string fixed for all subgraphs $H\in
\Lambda_\ell$, signifying that $H$ is a leaf node of $T$ and
$\textit{code}(H)$ can be obtained from $X_{\textit{base}}$ in $O(1)$
time using the position of $T_H$ in $T$.  If $k=|V(H)|>\ell$, then let
$(U_0,\ldots,U_p)$ and $(H_0,\ldots,H_p)$ be the $T$-partition and the
$T$-subgraphs of $H$.  To use the framework, one just has to provide
an $O(k)$-time obtainable string $\chi_{q,H_0}$ with
\[
\|\chi_{q,H_0}\|=o(k)+O(|U_0|\cdot\log k)
\]
such that the query $q$ can be answered in $O(1)$ time from the
prefixed concatenation $\chi_{q,H}$ of $\chi_{q,H_0},\ldots,
\chi_{q,H_p}$ via the labels $L_H(u)$ and other information of the
vertices $u$ given to the query algorithm.  The query $q$ can then be
supported in $O(1)$ time from $X_{\textit{base}}$ and the prefixed
concatenation $X_q$ of $\chi_{\textit{label}}$,
$\chi_{q,\textit{leaf}}$, and $\chi_{q,G}$.  Since the overall number
of vertices in all nodes of $T$ is $O(n)$, the strings $\chi_{q,G}$
and $X_q$ can be computed in $O(n)$ time.  We prove the following
equation for each subtree $T_H$ of $T$ by induction on the bounded
height of $T_H$:
\begin{equation}
\label{equation:equation13}
\|\chi_{q,H}\|=o(k)+O(1+|\partial H|\cdot \log k).
\end{equation}
If $k\leq \ell$, then $\|\chi_{q,H}\|=O(1)$ implies
Equation~\eqref{equation:equation13} (even if $k=O(1)$ and $|\partial
H|=0$).  The basis holds.  If $k>\ell$, then
Equation~\eqref{equation:equation4} and Condition~\ref{condition:S3}
of $(V_0,\ldots,V_p)$ imply
\begin{align*}
\|\chi_{q,H}\|
&=o(k)+O(|U_0|\cdot\log k)+
\sum_{i\in[p]}o(|V(H_i)|)+O(1+|\partial H_i|\cdot\log |V(H_i)|)\\
&=o(k)+O(1+|\partial H|\cdot\log k)
\end{align*}
by the inductive hypothesis. Equation~\eqref{equation:equation13} is
proved.  By $|\partial G|=0$, we have $\|\chi_{q,G}\|=o(n)$.

\subsection{Degree and other information of a vertex}
\label{subsection:subsection3.2}
For the query $q$ of obtaining the degree $|N_G(u)|$ of a vertex $u$
in $G$ from $L(u)$, let $\chi_{q,H_0}$ be the $O(k)$-time obtainable
$O(|U_0|\cdot\log k)$-bit string whose $L_H(u)$-th $\lceil\log
k\rceil$-bit word $\chi_{q,H_0}(L_H(u))$ stores the degree $|N_H(u)|$
of the vertex $u\in U_0$.  As instructed by the framework, it suffices
to show as follows that it takes $O(1)$ time to obtain $|N_G(u)|$ from
the prefixed concatenation $\chi_{q,H}$ of
$\chi_{q,H_0},\ldots,\chi_{q,H_p}$ via $L_H(u)$ for the case with
$k=V(H)>\ell$: Obtain the index $i\in [0,p]$ with $u\in U_i$ and
$L_{H_i}(u)$ from $\chi_{\textit{label}}$ via
Query~\ref{condition:L1}.  If $i=0$, then return
$|N_H(u)|=\chi_{q,H_0}(L_H(u))$ in $O(1)$ time.  If $i\in[p]$, then
return $|N_H(u)|=|N_{H_i}(u)|$ which can be obtained from
$\chi_{q,H_i}$ and $L_{H_i}(u)$ in $O(1)$ time. Note that this query
$q$ does not need Query~\ref{condition:L2}.

Other queries $q$ for a vertex $u$ of $G$ like (i) the equipped color
of $u$ in $G$, (ii) the number of incident outgoing edges of $u$ in
$G$, or (iii) the label $L(v)$ of an arbitrary neighbor $v$ of $u$ in
$G$ can be supported in a same way as long as (1) the answer can be
represented in $O(\log n)$ bits, (2) the answer to the query $q$ for a
vertex $u\in U_i$ with $i\in[0,p]$ can be obtained from merely
$\chi_{q,H_i}$ using $L_{H_i}(u)$, and (3) the answers for the
subgraphs $H$ of $G$ in $\Lambda_\ell$ can be stored in an $o(n)$-bit
string $\chi_{q,\textit{leaf}}$ with the help of $X_{\textit{base}}$
to provide $\textit{code}(H)$.

\subsection{Adjacency and other information between a pair of adjacent vertices}
\label{subsection:subsection3.3}
We first claim an $O(n)$-time obtainable $O(1)$-orientation $D$ for
$G$ that can be represented using only $o(n)$ bits into addition to
$X_\textit{base}$.  The query of determining the directions of the
edges between two vertices in $G$ via their labels can then be
supported in $O(1)$ time by the query $q$ for obtaining the labels
$L(v)$ and the directions of the edges of $G[\{u,v\}]$ for the $O(1)$
outgoing incident edges $\vec{uv}$ of $u$ in $D$ from $L(u)$: To
determine whether $\vec{uv}$ is an edge of $G$, just run the query
algorithm of $q$ on $u$ and $v$ to report the $O(1)$ outgoing edges of
$u$ and $v$ in $D$ and the directions of these $O(1)$ edges in $G$.
Since $D$ is an $O(1)$-orientation for $G$, $\vec{uv}$ is an edge of
$G$ if and only if it is detected by running the query algorithm for
$q$ on $u$ and $v$ in $O(1)$ time.

This query $q$ is about $O(\log n)$-bit representable information of a
vertex, but the approach of the previous subsection does not directly
work.  Since the answer involves labels of the neighbors of the
queried vertex, Query~\ref{condition:L2} has to come in: Let
$\chi_{q,H_0}$ be an $O(|U_0|\cdot\log k)$-bit string whose $j$-th
$O(\log k)$-bit word $\chi_{q,H_0}(j)$ with $j\in [|U_0|]$ stores the
label $L_H(v)$ and directions of the edges of $H[\{u,v\}]$ for each of
the $O(1)$ $u$-out edges $\vec{uv}$ in $D$ for the vertex $u\in U_0$
with $L_H(u)=j$.  To obtain the answer for a vertex $u\in V(H)$ from
$\chi_{q,H}$ via $L_H(u)$ for the case with $k=V(H)>\ell$, we also
obtain the index $i\in [0,p]$ with $u\in U_i$ and $L_{H_i}(u)$ from
$\chi_{\textit{label}}$ via Query~\ref{condition:L1}.  If $i=0$, then
return the answer stored in $\chi_{q,H_0}(L_H(u))$.  If $i\in [p]$,
then $|N_H(u)|=\chi_{q,H_0}(L_H(u))$ in $O(1)$ time.  If $i\in[p]$,
then we obtain the answer in $H_i$ from $\chi_{q,H_i}$ via
$L_{H_i}(u)$.  To return the answer of $u$ in $H$, we need
Query~\ref{condition:L2} to obtain the label $L_H(v)$ from
$\chi_{\textit{label}}$ via $L_{H_i}(v)$ and $i$ for each of the
$O(1)$ $u$-out edges $\vec{uv}$ in $D$.

It remains to prove the claim by (i) presenting an $O(k)$-time
obtainable $O(1)$-orientation $D_H$ of the $k$-vertex graph $H$ for
each subtree $T_H$ of $T$ and (ii) showing that $D=D_G$ can be
represented in $o(n)$ bits in addition to $X_{\textit{base}}$ to
support the $O(1)$-time query $d$ of reporting the $O(1)$ labels
$L(v)$ for the $u$-out edges $\vec{uv}$ in $D$ from $L(u)$.  By
$\ell=O(\log\log n)$, it takes $O(n)$ time to obtain an $o(n)$-bit
string $\chi_{d,\textit{leaf}}$ assigning an $O(1)$-orientation $D_H$
for each distinct graph $H\in \Lambda_\ell$ such that the query $q$
for each vertex $u$ of $H$ can be answered in $O(1)$ time from
$\chi_{d,\textit{leaf}}$ via $L_H(u)$ and $\textit{code}(H)$, which
can be obtained in $O(1)$ time from $X_{\textit{base}}$ using the
position of $T_H$ in $T$.  If $k>\ell$, then let $(U_0,\ldots,U_p)$
and $(H_0,\ldots,H_p)$ be the $T$-partition and the $T$-subgraphs of
$H$. For each $i\in[p]$, let $D_i$ consist of the outgoing incident
edges of $U_i$ in $D_{H_i}$.  Let $W_0$ consist of the vertices of
$U_0$ and the vertices $u$ that can be reached by exactly one outgoing
incident edges of $U_0$ in $D_{H_1}\cup\cdots\cup D_{H_p}$.  Let $D_0$
consist of the outgoing incident edges of $W_0$ in an arbitrary
$O(k)$-time obtainable $O(1)$-orientation $D'_H$ for $H$.  The
$O(k)$-time obtainable graph
\[
D_H=D_0\cup D_1\cup \cdots\cup D_p
\]
is an $O(1)$-orientation for $H$:\footnote{The following proof is a
specialized version (i.\,e., for the case with $t=1$) of the proof of
Lemma~\ref{lemma:lemma6} in~\S\ref{subsubsection:subsubsection3.4.2}.}
The number of $u$-out edges in $D_H$ of each vertex $u\in U_0$ is that
in $D_0\subseteq D'_H$ and hence is $O(1)$.  The number of $u$-out
edges in $D_H$ of each vertex $u\in U_i$ with $i\in [p]$ is exactly
that in $D_0\cup D_i\subseteq D'_H\cup D_{H_i}$ and hence is $O(1)$.
Assume for contradiction $E(D_H[\{u,v\}])=\varnothing$ for adjacent
vertices $u\in U_i$ and $v\in U_j$ in $H$ with $\{i,j\}\subseteq
[0,p]$.  We have $i\ne j$ or else $u$ and $v$ are adjacent in
$D_i=D_j$.  We have $ij=0$ or else $u$ and $v$ are non-adjacent in
$H$.  Let $i=0$ and $j\in [p]$ without loss of generality. We have
$\vec{uv}\notin E(D'_H)$ or else $\vec{uv}\in E(D_0)\subseteq E(D_H)$
by $u\in U_0\subseteq W_0$.  Since $D'_H$ is an $O(1)$-orientation for
$H$, we have $\vec{vu}\in E(D'_H)$.  We also have $\vec{uv}\notin
E(D_{H_j})$ or else $v\in W_0$ implies $\vec{vu}\in E(D_0)\subseteq
E(D_H)$, contradicting $E(D_H[\{u,v\}])=\varnothing$.  Since $D_{H_j}$
is an $O(1)$-orientation for $H_j$, we have $\vec{vu}\in E(D_{H_j})$,
implying $\vec{vu}\in E(D_j)\subseteq E(D_H)$, contradicting
$E(D_H[\{u,v\}])=\varnothing$.  Thus, $D_H$ is indeed an
$O(1)$-orientation for $H$.

To see that $D=D_G$ can be represented in $o(n)$ bits in addition to
$X_{\textit{base}}$ to support the query $d$, let $\chi_{q,H_0}$ be
the $o(k)+O(|U_0|\cdot\log k)$-bit prefixed concatenation of the
following strings $\chi_1$ and $\chi_2$.  Observe first that each
$D_{H_i}$ with $i\in[p]$ is an $O(1)$-orientation for $H_i$, implying
\[
|W_0|=|U_0|+\sum_{i\in[p]}O(|N_H(U_i)|)=O(|U_0|)+o(k/\log k)
\]
by Condition~\ref{condition:S3} of the corresponding star partition 
$(V_0,\ldots,V_p)$ and Equation~\eqref{equation:equation4}.
\begin{itemize}
\item 
Let $\chi_1=\textit{dict}(Y)$ for the $k$-bit string $Y$ such that
$Y[L(u)]=1$ if and only if $u\in W_0$.  By Lemma~\ref{lemma:lemma4},
$\|\chi_1\|=o(k)+O(|W_0|\cdot\log k)=o(k)+O(|U_0|\cdot\log k)$.

\item 
Let $\chi_2$ be a string whose $j$-th $O(\log k)$-bit word $\chi_2(j)$
for each $j\in [|W_0|]$ stores $L_H(v)$ and $H[\{u,v\}]$ of each
$u$-out edge $\vec{uv}$ for the vertex $u$ with
$\textit{rank}(Y,L_H(u))=j$ in $D_0$.  We have
$\|\chi_2\|=O(|W_0|\cdot \log k)=o(k)+O(|U_0|\cdot\log k)$.
\end{itemize}
The query $d$ for a vertex $u\in V(H)$ can be supported in $O(1)$ time
from $\chi_{d,H}$ via $L_H(u)$ as follows: Obtain the index $i\in
[0,p]$ with $u\in U_i$ and $L_{H_i}(u)$.  If $Y[L_H(u)]=1$, then
report the information stored in $\chi_2(\textit{rank}(Y,L_H(u)))$ in
$O(1)$ time first.  If $i\in[p]$, then obtain the labels $L_{H_i}(v)$
using the query $d$ in $H_i$ via $L_{H_i}(u)$ in $O(1)$ time and
return the labels $L_H(v)$ using Query~\ref{condition:L2}.  According
to the framework of~\S\ref{subsection:subsection3.1}, the query $d$ is
supported in $O(1)$ time using an $o(n)$-bit $X_d$ and
$X_{\textit{base}}$.

Other queries $q$ for a pair of adjacent vertices $u$ and $v$ of $G$
can be supported in a same way as long as (1) the answer can be
represented in $O(\log n)$ bits and (2) the answers for the subgraphs
of $G$ in $\Lambda_\ell$ can be stored in overall $o(n)$ bits in
addition to $X_{\textit{base}}$.  For example, our encoding scheme
$A^*$ supports in $O(1)$ time the query $q$ of reporting the the
neighbor of $u$ succeeding (respectively, preceding) $v$ in clockwise
order around $u$ according to the embedding of $G$ reflected by the
adjacency list of $G$ decoded by~$A^*$.  Specifically, let
$\chi_{q,H_0}$ additionally store for each $u$-out edge $\vec{uv}$ of
$D$ the $O(\log |V(H)|)$-bit information of the $O(1)$ triples
$(u,v,w)$ or $(v,u,w)$ stored in $\textit{code}(H_0)$.  Since the bit
count of $\chi_{q,H_0}$ remains $o(k)+O(|U_0|\cdot\log k)$, the query
$q$ can be supported in $O(1)$ time via $X_{\textit{base}}$ and an
$o(n)$-bit string $X_q$.  Combing with the $O(1)$-time query of
reporting an arbitrary neighbor $v$ of $u$, the query of listing all
neighbors of a vertex $u$ can be supported in $O(1)$ time per output.

\subsection{Bounded-distance shortest path}
\label{subsection:subsection3.4}

\begin{figure}
\centerline{\scalebox{0.5}{\includegraphics{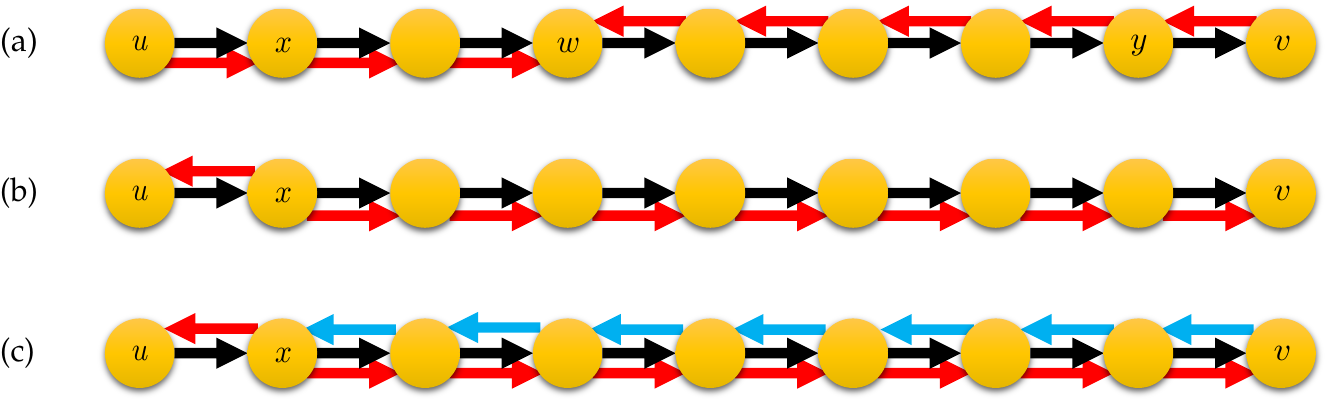}}}
\caption{Let $d_G(u,v)=8$.  The red edges belong to a $10$-director
  $D$ for $G$ with black edges.  The blue edges belong to an
  $8$-enhancer $E_8$ for $D$.  Let $D_8=D\cup E_8$.  (a) A
  $D$-directive $uv$-path of $G$.  (b) A $D$-pivoted $uv$-path of
  $G$. (c) A shortest $uv$-path $P$ of $G$ with $P-u\subseteq D_8\cap
  D_8^r$.}
\label{figure:figure4}
\end{figure}

Let $d_G(u,v)$ for vertices $u$ and $v$ of a graph $G$ denote the
length of a shortest $uv$-path of $G$.  A $uv$-path $P$ of $G$ is
\emph{$D$-directive} for a graph $D$ if $P$ contains a vertex $w$ with
\[
d_{P\cap D}(u,w)+d_{P\cap D^r}(w,v)=d_G(u,v).
\]
A $D$-directive $uv$-path of $G$ is a shortest $uv$-path of $G$ but
not vice versa.  A vertex pair $(u,v)$ of $G$ is \emph{$D$-directive}
if $G$ contains a $D$-directive $uv$-path.  For a positive integer
$t$, a \emph{$t$-director} for $G$ is an $O(1)$-orientation $D$ for
$G$ such that each vertex pair $(u,v)$ of $G$ with $d_G(u,v)\in [t]$
is $D$-directive.  Thus, each pair $(u,v)$ with $d_G(u,v)\in [t]$
admits an edge $\vec{ux}$ of $G\cap D$ with $d_G(x,v)=d_G(u,v)-1$ or
an edge $\vec{yv}$ of $G\cap D^r$ with $d_G(u,y)=d_G(u,v)-1$.  See
Figure~\ref{figure:figure4}(a) for an illustration.  An
$O(1)$-orientation for $G$ is a $1$-director for~$G$ and vice versa.
We prove the following lemma
in~\S\ref{subsubsection:subsubsection3.4.1}.

\begin{lemma}
\label{lemma:lemma5}
For any prespecified positive integer $t=O(1)$, it takes $O(n)$ time
to obtain a $t$-director for an $n$-vertex graph $G$ with
$\eta(G)=O(1)$.
\end{lemma}

As mentioned in the technical overview of~\S\ref{section:section1},
Lemma~\ref{lemma:lemma5} and our $\cF$-optimal encoding scheme $A^*$
immediately lead to an $O(n)$-time obtainable $O(n)$-bit encoded
string for any graph $G$ with $\eta(G)=O(1)$ equipped with a
$t$-director $D$, supporting the query $q$ of bounded-distance
shortest path in $O(1)$ time, already improving upon the $O(n\log
n)$-bit data structure of Kowalik and Kurowski~\cite{KowalikK06}: To
obtain a shortest $uv$-path of $G$ for vertices $u$ and $v$ with
$d_G(u,v)\leq t$, just find a shortest $uv$-path in the $O(1)$-vertex
subgraph of $G$ induced by
\[
W=\{w\in V(G): \min\{d_D(u,w),d_D(v,w)\}\leq t\},
\] 
which can be obtained from accessing the $O(1)$-time query $d$ of
obtaining the $O(1)$ $u$-out edges $\vec{uv}$ in $D$ and their
$G[\{u,v\}]$ for $O(1)$ times due to $t=O(1)$.
 
An $\cF$-succinct encoded string for $G$ equipped with an arbitrary
$t$-director for $G$ need not be $\cF$-succinct for $G$. To support
the query $q$ of bounded-distance shortest path in $O(1)$ time without
affecting the $\cF$-succinctness of the encoded string for $G$, we
present an $O(n)$-time obtainable $t$-director $D=D_G$ for $G$ which
can be represented using an $o(n)$-bit string $X_d$ in addition to
$X_{\textit{base}}$ supporting the above query $d$ in $O(1)$ time.
Specifically, we define for each subtree $T_H$ of $T$ a $t$-director
$D_H$ for $H$ by the following recursive procedure: We apply
Lemma~\ref{lemma:lemma5} to obtain a $t$-director $D_H$ for each
distinct subgraph $H$ of $G$ in $\Lambda_\ell$.  By $\ell=O(\log\log
n)$, it takes $O(n)$ time to obtain an $o(n)$-bit string
$\chi_{d,\textit{leaf}}$ such that the answer to the query $d$ for a
vertex $u\in V(H)$ can be obtained in $O(1)$ time from
$\chi_{d,\textit{leaf}}$ via $L_H(u)$ and $\textit{code}(H)$, which is
$O(1)$-time obtainable from $X_{\textit{base}}$ using the position of
the subtree $T_H$ in $T$.  For a subtree $T_H$ of $T$ with
$|V(H)|>\ell$, let $(U_0,\ldots,U_p)$ and $(H_0,\ldots,H_p)$ be the
$T$-partition and $T$-subgraphs of $H$.  Let each $D_i$ with $i\in[p]$
consist of the $u$-out edges for all vertices $u\in U_i$ in the
recursively defined $t$-director $D_{H_i}$ for $H_i$.  Let each $W_i$
with $i\in[p]$ consist of the vertices $u\in U_i$ with
$d_{D_{H_i}}(U_0,u)\leq t$. That is, $u\in W_i$ if and only if there
is an $xu$-path in $D_{H_i}$ of length at most $t$ for a vertex $x\in
U_0$.  Let $D_0$ consist of the $u$-out edges for all vertices $u$ in
\[
W_0=U_0\cup W_1\cup \cdots \cup W_p
\]
in an arbitrary $O(k)$-time obtainable $t$-director $D'_H$ for $H$. We
prove the following lemma in~\S\ref{subsubsection:subsubsection3.4.2}.
\begin{lemma}
\label{lemma:lemma6}
$D_H=D_0\cup D_1\cup \cdots\cup D_p$ is a $t$-director for $H$.
\end{lemma}
Since the overall number of vertices in all nodes of $T$ is $O(n)$,
the $t$-director $D=D_G$ for $G$ can be obtained in $O(n)$ time.
Since $t=O(1)$ and each $D_{H_i}$ with $i\in[p]$ is an
$O(1)$-orientation for $H_i$, we have
\[
|W_0| =|U_0|+\sum_{i\in[p]}O(|N_H(U_i)|)=O(|U_0|)+o(k/\log k)
\]
by Condition~\ref{condition:S3} of the corresponding star partition
$(V_0,\ldots,V_p)$ of $H$ and Equation~\eqref{equation:equation4}.
Thus, the query $d$ can be supported by an $o(n)$-bit string $X_d$ and
$X_{\textit{base}}$ in $O(1)$ time in precisely the same way as the
query $d$ for the $O(1)$-orientation $D$ for $G$
in~\S\ref{subsection:subsection3.3}.  The rest of the subsection
proves Lemmas~\ref{lemma:lemma5} and~\ref{lemma:lemma6}
in~\S\ref{subsubsection:subsubsection3.4.1}
and~\S\ref{subsubsection:subsubsection3.4.2}, respectively.

\subsubsection{Proving Lemma~\ref{lemma:lemma5}}
\label{subsubsection:subsubsection3.4.1}

\begin{lemma}[see, e.g.,~\cite{Kostochka84,Mader67,Thomason01}]
\label{lemma:lemma7}
If $\eta(G)=O(1)$, then $\min\{|N_G(v)|:v\in V(G)\}=O(1)$.
\end{lemma}

\begin{lemma}
\label{lemma:lemma8}
It takes $O(n)$ time to expand a $t$-director for an $n$-vertex graph
$G$ with $t=O(1)$ and $\eta(G)=O(1)$ into a $t+1$-director for $G$.
\end{lemma}

We first reduce Lemma~\ref{lemma:lemma5} to Lemma~\ref{lemma:lemma8}
via Lemma~\ref{lemma:lemma7} and then prove Lemma~\ref{lemma:lemma8}.

\begin{proof}[Proof of Lemma~\ref{lemma:lemma5}]
By $t=O(1)$ and Lemma~\ref{lemma:lemma8}, it suffices to show that a
$1$-director $D$ for the $n$-vertex input graph $G$ with
$\eta(G)=O(1)$ can be obtained in $O(n)$ time: For each $i\in [n]$,
let $u_i$ be a vertex of $G_i=G-\{u_j:j\in [i-1]\}$ with minimum
$|N_G(u_i)|$, which is $O(1)$ by Lemma~\ref{lemma:lemma7}.  It takes
$O(n)$ time to obtain the $O(1)$-orientation $D=\{\vec{u_iv_i}: v_i\in
N_{G_i}(u_i), i\in[n]\}$ for $G$.
\end{proof}

\begin{proof}[Proof of Lemma~\ref{lemma:lemma8}]
Let $D$ be a $t$-director for $G$.  
A $uv$-path $P$ of $G$ with $|E(P)|\geq 1$ is
\emph{$D$-pivoted} if
\[
d_{P\cap D^r}(u,x)+d_{P\cap D}(x,v)=d_G(u,v)\geq 1
\] 
holds for the neighbor $x$ of $u$ in $P$.  See
Figure~\ref{figure:figure4}(b) for an illustration.  A $D$-pivoted
$uv$-path of $G$ is a shortest $uv$-path of $G$ but not vice versa.  A
vertex pair $(u,v)$ of $G$ is \emph{$D$-pivoted} if $G$ contains a
$D$-pivoted $uv$-path.  A graph $E_j$ is a \emph{$j$-enhancer} for $D$
with $j\in [t+1]$ if the next \emph{Conditions~E} hold:
\begin{enumerate}[label=\emph{E\arabic*:}, ref={E\arabic*}]
\item 
\label{condition:E1}
$D_j=D\cup E_j$ remains an $O(1)$-orientation for $G$ and hence a
$t$-director for $G$.

\item 
\label{condition:E2}
Each $D$-pivoted vertex pair $(u,v)$ of $G$ with $d_G(u,v)\in [j]$
such that $G\cap D_j$ does not contain any edge $\vec{ux}$ with
$d_G(x,v)=j-1$ admits a $D$-pivoted $uv$-path $P$ of $G$ with
$P-u\subseteq D_j\cap D_j^r$.
\end{enumerate}
See Figure~\ref{figure:figure4}(c)
for an illustration. 

Claim~1: It takes $O(n)$ time to expand a $j$-enhancer for $D$ with
$j\in[t]$ into a $j+1$-enhancer for $D$.

By Claim~1 and $t=O(1)$, it suffices to ensure that $D_{t+1}=D\cup
E_{t+1}$ for a $t+1$-enhancer for $D$ is a $t+1$-director for $G$,
since the empty graph is a $1$-enhancer for $D$.

Assume for contradiction that a vertex pair $(u,v)$ of $G$ with
$d_G(u,v)\in [t+1]$ is not $D_{t+1}$-directive.  By
Condition~\ref{condition:E1} for $E_{t+1}$, we have $d_G(u,v)=t+1$,
implying an edge $\vec{ux}$ of $G$ with $d_G(x,v)=t$.  $D$ does not
contain $\vec{ux}$ or else the union of $\vec{ux}$ and a $D$-directive
$xv$-path of $G$ is a $D_{t+1}$-directive $uv$-path of $G$.  Since $D$
is an orientation for $G$ that does not contain $\vec{ux}$, $D^r$
contains $\vec{ux}$.  Let $Q$ be a $D$-directive $xv$-path of $G$.
$D$ contains $Q$ or else the union of a $D$-directive $uy$-path for
the neighbor $y$ of $v$ in $Q$ and the edge $\vec{yv}$ is a
$D_{t+1}$-directive $uv$-path of $G$.  Hence, $\vec{ux}\cup Q$ is a
$D$-pivoted $uv$-path of $G$, implying that $(u,v)$ is a $D$-pivoted
vertex pair of $G$ with $d_G(u,v)=t+1$.  $G\cap D_{t+1}$ does not
contain any edge $\vec{ux}$ with $d_G(x,v)=t$ or else the union of
$\vec{ux}$ and a $D$-directive $xv$-path of $G$ is a
$D_{t+1}$-directive $uv$-path of $G$.  By Condition~\ref{condition:E2}
for $E_{t+1}$, $G$ contains a $D$-pivoted $uv$-path $P$ with
$P-u\subseteq D_{t+1}\cap D_{t+1}^r$, implying that $P$ is a
$D_{t+1}$-directive $uv$-path of $G$, contradiction.

The rest of the proof ensure Claim~1 by showing a graph $E$ obtainable
in $O(n)$ time from a $j$-enhancer $E_j$ for $D$ such that
$E_{j+1}=E_j\cup E$ is a $j+1$-enhancer for $D$.  We rely on the
following graph $H$ and an $O(1)$-orientation $D_H$ for $H$: Let $H$
be the graph on $V(G)$ consisting of the edges $\vec{uv}$ such that
$(u,v)$ are vertex pairs of $G$ with $d_G(u,v)=j+1$ that admit
$D$-pivoted $uv$-paths $P$ of $G$ with $P-\{u,v\}\subseteq D_j\cap
D_j^r$.  We show that $H$ can be obtained in $O(n)$ time.  Since $D$
is a $t$-director for $G$ with $j\leq t=O(1)$, it takes $O(1)$ time to
determine for each vertex pair $(u,v)$ of $G$ whether $d_G(u,v)\geq
j+1$.  Since $D$ is an $O(1)$-orientation, each vertex $x$ of $G$ can
be the second vertex of $O(1)$ length-$j+1$ paths of $G$ that are
$D_j$-pivoted.  Thus, it takes $O(n)$ time to obtain the set $\cP$ of
all $D$-pivoted $uv$-paths $P$ of $G$ satisfying $d_G(u,v)=j+1$ and
$P-\{u,v\}\subseteq D_j\cap D_j^r$. We have $|\cP|=O(n)$.  Since $H$
consists of the edges $\vec{uv}$ admitting $uv$-paths in $\cP$, it
takes $O(n)$ time to construct $H$.

Claim~2: It takes $O(n)$ time to obtain an $O(1)$-orientation $D_H$ for $H$.

\begin{figure}
\centerline{\scalebox{0.5}{\includegraphics{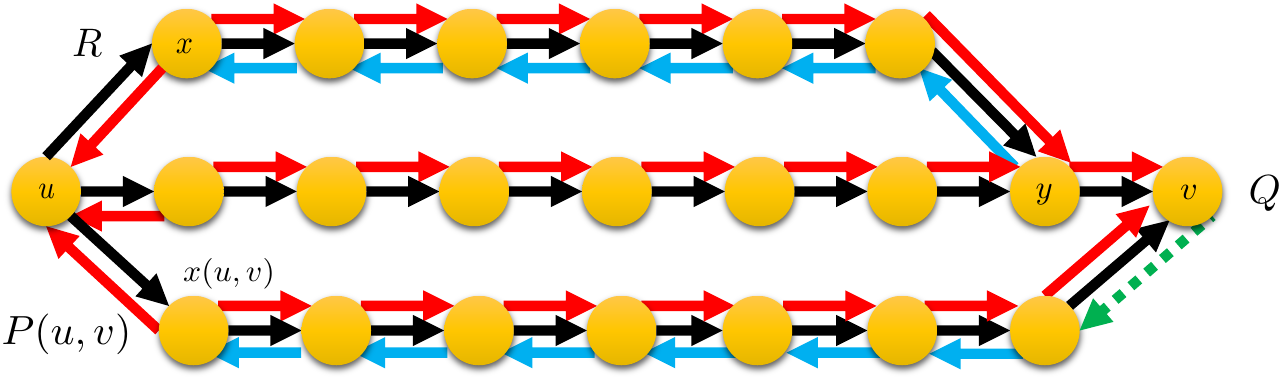}}}
\caption{An illustration for the proof of Lemma~\ref{lemma:lemma8}.
  The black, red, blue, and dotted edges are in $G$, $D$, $E_j$, and
  $E$, respectively.  $Q$ is the $D$-pivoted $uv$-path of $G$ in the
  middle.  $R$ is the $D$-pivoted $uy$-path of $G$ above $Q$.
  $P(u,v)$ is the $D$-pivoted $uv$-path of $G$ is below $Q$.}
\label{figure:figure5}
\end{figure}

For each edge $\vec{uv}$ of $H$, let $P(u,v)$ be an arbitrary fixed
$uv$-path in $\cP$ and let $x(u,v)$ be the neighbor of $u$ in
$P(u,v)$.  Define $E$ as the $O(n)$-time obtainable subgraph of $G\cup
G^r$ consisting of
\begin{itemize}
\item 
the incident edge of $u$ in $P(u,v)$ for each edge $\vec{uv}$ of
$H\cap D_H$ and

\item 
the incident edge of $v$ in $(P(u,v))^r$ for each edge $\vec{uv}$ of $H\cap
D_H^r$.
\end{itemize}
Since $D_H$ is an $O(1)$-orientation, so is $E$, implying that
$D_{j+1}=D_j\cup E$ remains an $O(1)$-orientation for
$G$. Condition~\ref{condition:E1} for $E_{j+1}$ holds.  To see
Condition~\ref{condition:E2} for $E_{j+1}$, let $(u,v)$ be a
$D$-pivoted vertex pair of $G$ with $d_G(u,v)\in [j+1]$ such that
$G\cap D_{j+1}$ does not contain any edge $\vec{ux}$ with
$d_G(x,v)=j$.  By Condition~\ref{condition:E2} for $E_j$ and
$E_j\subseteq E_{j+1}$, it suffices to focus on the case with
$d_G(u,v)=j+1$.  Let $Q$ be a $D$-pivoted $uv$-path of $G$.  See
Figure~\ref{figure:figure5}.  Let $\vec{yv}$ be the incident edge of
$v$ in $Q$.  Since $Q-v$ is a $D$-pivoted $uy$-path of $G$, $(u,y)$ is
a $D$-pivoted vertex pair of $G$ with $d_G(u,y)=j$.  $G\cap D_j$ does
not contain any edge $\vec{ux}$ with $d_G(x,y)=j-1$ or else $\vec{ux}$
is an edge of $G\cap D_{j+1}$ with $d_G(x,v)=j$.  By
Condition~\ref{condition:E2} for $E_j$, there is a $D$-pivoted
$uy$-path $R$ of $G$ such that $D_j\cap D_j^r$ contains $R-u$ or
$R-y$.  Observe that $D_j\cap D_j^r$ does not contain $R-y$ or else
the incident edge $\vec{ux}$ of $u$ in $R$ is an edge of $G\cap D_j$
with $d_G(x,y)=j-1$.  Thus, $D_j\cap D_j^r$ contains $R-u$, implying
that the union $P$ of the $D$-pivoted $uy$-path $R$ and the edge
$\vec{yv}$ of $G\cap D$ is a $D$-pivoted $uv$-path of $G$ with
$P-\{u,v\}\subseteq D_j\cap D_j^r$.  Hence, $\vec{uv}$ is an edge of
$H$.  By $d_G(x(u,v),v)=j$, the incident edge of $u$ in $P(u,v)$ is
not in $G\cap D_{j+1}$.  Thus, $E$ contains the incident edge of $v$
in $(P(u,v))^r$, implying that $P(u,v)$ is a $D$-pivoted $uv$-path of
$G$ with $P(u,v)-u\subseteq D_{j+1}\cap D_{j+1}^r$.
Condition~\ref{condition:E2} for $E_{j+1}$ holds.

It remains to prove Claim~2.  We need an $O(n)$-time obtainable
$O(1)$-coloring of the graph $F=D_j^{2j-1}$ to construct $D_H$.  That
is, $F$ consists of the edges $\vec{uv}$ with $d_{D_j}(u,v)\in
[2j-1]$.  Since $D_j$ is an $O(1)$-orientation, so is $F$ by $j\leq
t=O(1)$.  Let each $v_i$ with $i\in[n]$ be an arbitrary vertex of
$F_i=F-\{v_1,\ldots,v_{i-1}\}$ that minimizes $|N_{F_i}(v_i)|$. Since
each $F_i$ with $i\in[n]$ is an $O(1)$-orientation, we have
$|N_{F_i}(v_i)|=O(1)$. Thus, $F$ admits an $O(n)$-time obtainable
$O(1)$-coloring via assigning for each $v_i$ with $i\in[n]$ a color
distinct from those of vertices in $N_{F_i}(v_i)$.  Based on the
$O(1)$-coloring of $F$, it takes $O(n)$ time to decompose $H$ into
$O(1)$ edge-disjoint subgraphs $H_k$ of $H$ whose union is $H$ such
that the vertices $x(u,v)$ for distinct edges $\vec{uv}$ in the same
subgraph $H_k$ are distinct vertices having the same color in $F$.  If
$\vec{uv}$ and $\vec{yz}$ are distinct edges of the same subgraph
$H_k$ of $H$, then $P(u,v)$ does not intersect $V(P(y,z))\setminus
\{y,z\}$ or else $d_{D_j}(x(u,v),x(y,z))\in[2j-1]$ contradicts that
$x(u,v)$ and $x(y,z)$ have the same color in $F$.  Hence, each $H_k$
can be obtained from the union of the paths $P(u,v)$ of $G$ over all
edges $\vec{uv}\in E(H_k)$ by contracting each $P(u,v)$ into
$\vec{uv}$, implying that $H_k$ is a minor of $G$.  By $\eta(H_k)\leq
\eta(G)=O(1)$, it takes $O(n)$ time to obtain an $O(1)$-orientation
$D_{H_k}$ for $H_k$.  The union $D_H$ of the $O(1)$ graphs $D_{H_k}$
is an $O(n)$-time obtainable $O(1)$-orientation for $H$.
\end{proof}

\subsubsection{Proving Lemma~\ref{lemma:lemma6}}
\label{subsubsection:subsubsection3.4.2}

\begin{figure}
\centerline{\scalebox{0.5}{\includegraphics{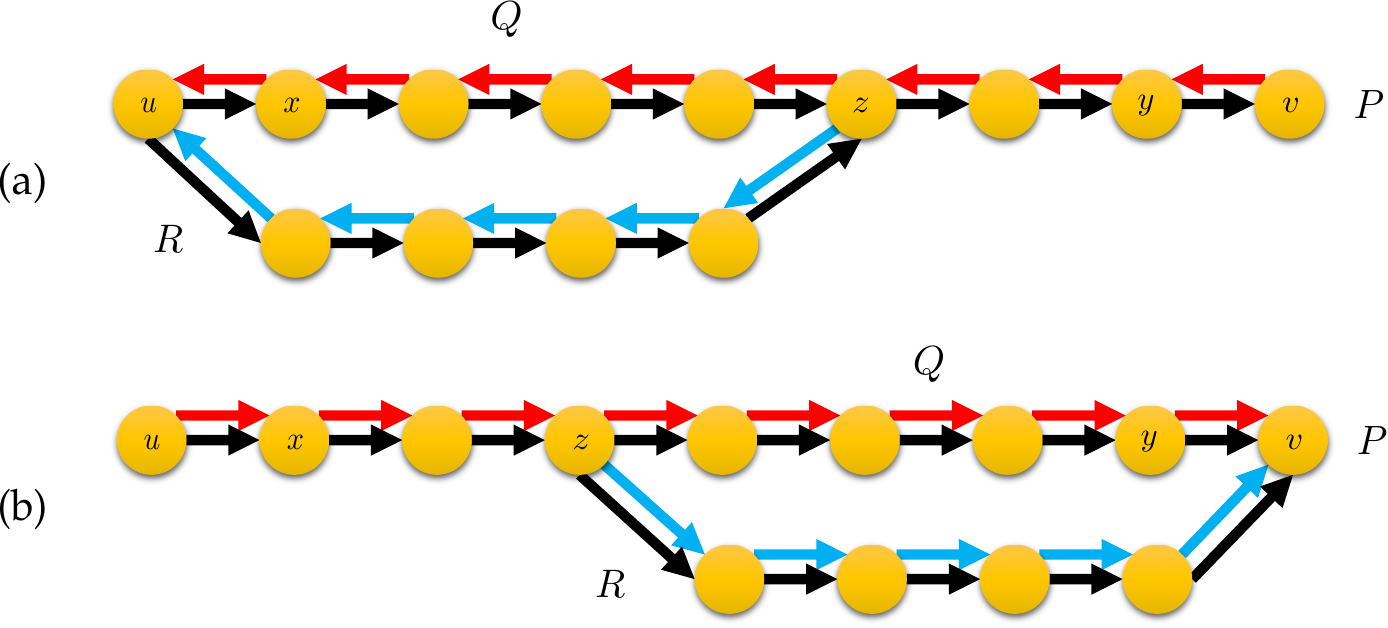}}}
\caption{An illustration for the proof of Statements~(a) and (b) in
  the proof of Lemma~\ref{lemma:lemma6}. The black edges belong to
  $H$. The red edges belong to $D'_H$.  $P$ is the $D'_H$-directive
  $uv$-path of $G$.  (a) The blue edges belong to $D_{H_i}$.  $Q$ is
  the $uz$-path of $P$. $R$ is the $D_{H_i}$-directive $uz$-path of
  $H_i$ below $P$.  (b) The blue edges belong to $D_{H_j}$.  $Q$ is
  the $zv$-path of $P$. $R$ is the $D_{H_j}$-directive $zv$-path of
  $H_j$ below $P$.}
\label{figure:figure6}
\end{figure}

\begin{proof}[Proof of Lemma~\ref{lemma:lemma6}]
Observe first that $D_H$ is an $O(1)$-orientation, since
$(U_0,\ldots,U_p)$ is a partition of $H$ and the $u$-out edges in
$D_H=D_0\cup \cdots \cup D_p$ for each vertex $u\in U_i$ with
$i\in[0,p]$ are $u$-out edges in the two $O(1)$-orientations $D'_H$
and $D_{H_i}$.  We show $D_H\cup D_H^r=H\cup H^r$ to further ensure
that $D_H$ is an $O(1)$-orientation for $H$.  Since each edge of $D_H$
belongs to an $O(1)$-orientation for a subgraph of $H$, we have
$D_H\subseteq H\cup H^r$.  Each edge of $H[U_i]$ with $i\in [0,p]$ is
in $D_i\cup D_i^r\subseteq D_H\cup D_H^r$.  An edge of $H$ not in any
$H[U_i]$ with $i\in[0,p]$ has to be an edge $\vec{uv}$ or $\vec{vu}$
of $H_i$ with $i\in [p]$ between a vertex $u\in U_i$ and a vertex
$v\in U_0\subseteq W_0$.  If $u\in W_0$, then $D'_H[\{u,v\}]\subseteq
D_0\subseteq D_H$.  If $u\notin W_0$, then $D_{H_i}$ does not contain
the edge $\vec{vu}$. Since $D_{H_i}$ is an orientation for $H_i$, we
have $\vec{uv}\in E(D_{H_i})$, implying that $\vec{uv}$ is an edge of
$D_i\subseteq D_H$.  Hence, $H\subseteq D_H\cup D_H^r$. Therefore,
$D_H$ is an $O(1)$-orientation for $H$.

It remains shows that each vertex pair $(u,v)$ of $H$ with
$d_H(u,v)\in [t]$ is $D_H$-directive.  Assume for contradiction that
$(u,v)$ is a non-$D_H$-directive pair of vertices $u\in U_i$ and $v\in
U_j$ with $i\in [0,p]$ that minimizes $d_H(u,v)\in [t]$.  The
following \emph{Conditions~D} follow from the minimality of
$d_H(u,v)$:
\begin{enumerate}[label=\emph{D\arabic*:}, ref={D\arabic*}]
\item 
\label{condition:D1}
$H\cap D_H$ contains no edge $\vec{uw}$ with $d_H(w,v)=d_H(u,v)-1$ or
else the union of $\vec{uw}$ and a $D_H$-directive $wv$-path of $H$ is
a $D_H$-directive $uv$-path of $H$.

\item 
\label{condition:D2}
$H\cap D_H^r$ contains no edge $\vec{wv}$ with $d_H(u,w)=d_H(u,v)-1$
or else the union of a $D_H$-directive $uw$-path of $H$ and $\vec{wv}$
is a $D_H$-directive $uv$-path of $H$.
\end{enumerate}
Let $P$ be a $D'_H$-directive $uv$-path of $H$, implying that $D'_H$
contains the incident edge $\vec{ux}$ of $u$ in $P$ or the incident
edge $\vec{vy}$ of $v$ in $P^r$.  If $i=j=0$, then $D_0\subseteq D_H$
contains $\vec{ux}$ or $\vec{vy}$, violating
Condition~\ref{condition:D1} with $w=x$ or
Condition~\ref{condition:D2} with $w=y$.  Thus, $i\in [p]$ or $j\in
[p]$.  However, we show as follows that (a) $i\in[p]$ implies
$j\in[p]$ and $\vec{ux}\notin E(D'_H)$ and (b) $j\in [p]$ implies
$i\in[p]$ and $\vec{vy}\notin E(D'_H)$.  This leads to $E(D'_H)\cap
\{\vec{ux},\vec{yv}\}=\varnothing$, contradiction.
 
Statement~(a): $i\in [p]$ implies $j\in[p]$ and $\vec{ux}\notin
E(D'_H)$.  Let the $uz$-path $Q$ be the maximal prefix of $P$ with
$Q\subseteq H_i$.  Let $R$ be a $D_{H_i}$-directive $uz$-path of
$H_i$. See Figure~\ref{figure:figure6}(a).  We have $R^r\subseteq
D_{H_i}$ or else $D_i\subseteq D_H$ contains the incident edge
$\vec{uw}$ of $u$ in $R$, violating Condition~\ref{condition:D1}.  We
have $z\in U_0$: If $z$ were not in $U_0$, then we have $v=z\in U_i$
by the maximality of $Q$, implying $P=Q\subseteq H_i$ and
$|E(R)|=d_{H_i}(u,v)=d_H(u,v)\in [t]$.  By $R^r\subseteq D_{H_i}$, the
incident edge $\vec{vw}$ of $v$ in $R^r$ is in $D_i\subseteq D_H$,
violating Condition~\ref{condition:D2}.  Thus, $z\in U_0$. By
$R^r\subseteq D_{H_i}$, we have $u\in W_0$, implying $\vec{ux}\notin
E(D'_H)$ or else $\vec{ux}\in E(D_0)\subseteq E(D_H)$ violates
Condition~\ref{condition:D1} with $w=x$.  By $\vec{ux}\notin E(D'_H)$,
we have $P^r\subseteq D'_H$. Thus, $j\ne 0$ or else $\vec{vy}\in
E(D'_H)$ implies $\vec{vy}\subseteq E(D_0)\subseteq E(D_H)$, violating
Condition~\ref{condition:D2} with $w=y$.

Statement~(b): $j\in [p]$ implies $i\in[p]$ and $\vec{vy}\notin
E(D'_H)$.  Let the $zv$-path $Q$ be the maximal suffix of $P$ with
$Q\subseteq H_j$.  Let $R$ be a $D_{H_j}$-directive $zv$-path of
$H_j$.  See Figure~\ref{figure:figure6}(b).  We have $R\subseteq
D_{H_j}$ or else $D_j\subseteq D_H$ contains the incident edge
$\vec{vw}$ of $v$ in $R^r$, violating Condition~\ref{condition:D2}.
We have $z\in U_0$: If $z$ were not in $U_0$, then we have $u=z\in
U_j$ by the maximality of $Q$, implying $P=Q\subseteq H_j$ and
$|E(R)|=d_{H_j}(u,v)=d_H(u,v)\in[t]$. By $R\subseteq D_{H_j}$, the
incident edge $\vec{uw}$ of $u$ in $R$ is in $D_j\subseteq D_H$,
violating Condition~\ref{condition:D1}. Thus, $z\in U_0$. By
$R\subseteq D_{H_j}$, we have $v\in W_0$, implying $\vec{vy}\notin
E(D'_H)$ or else $\vec{vy}\in E(D_0)\subseteq E(D_H)$ violates
Condition~\ref{condition:D2} with $w=y$.  By $\vec{vy}\notin E(D'_H)$,
we have $P\subseteq D'_H$. Thus, $i\ne 0$ or else $\vec{ux}\in
E(D'_H)$ implies $\vec{ux}\in E(D_0)\subseteq E(D_H)$ violates
Condition~\ref{condition:D1} with $w=x$.
\end{proof}

\section{Concluding remarks}
\label{section:section4}
We propose to base an encoding scheme on multiple classes of graphs to
exploit the inherent structures of individual graphs. As the first
nontrivial such example, we present an $\cF$-optimal encoding scheme
$A^*$ for a family $\cF$ of an infinite number of classes.
Specifically, $A^*$ takes an $n$-vertex $k$-clique-minor-free graph
$G$ with $k=O(1)$ and produces a $\cC$-succinct encoded string $X$ in
deterministic $O(n)$ time for each nontrivial quasi-monotone class
$\cC$ of graphs that contains $G$.  $A^*$ can decode $X$ back to $G$
in deterministic $O(n)$ time. This means that $A^*$ automatically
exploits an infinite number of possible nontrivial quasi-monotone
structures of $G$ to encode $G$ as compactly as possible. $A^*$ does
not require any embedding of $G$ or any recognition algorithm or other
explicit or implicit knowledge about the member classes $\cC$ of
$\cF$.  $A^*$ also supports fundamental queries in $O(1)$ time per
output.  Moreover, $A^*$ accepts additional information like an
$O(1)$-coloring or an genus-$O(1)$ embedding of $G$ that can be
decoded back together with $G$ by $A^*$ and answered by the query
algorithms of $A^*$.  It is of interest to see if our $\cF$-optimal
encoding scheme $A^*$ can be extended so that (i) the ground $\bigcup
\cF$ can be a proper superclass of the graphs having bounded Hadwiger
numbers, (ii) the member classes $\cC$ of $\cF$ can be more refined to
beyond quasi-monotonicty, or (iii) efficient updates can be supported
for the input graph $G$ and its equipped information.

\bibliographystyle{hilabbrv}
\bibliography{slim}
\end{document}